\documentclass{article}

\usepackage{microtype}
\usepackage{graphicx}
\usepackage{subfigure}
\usepackage{booktabs} 

\usepackage{natbib}
\usepackage{fullpage}

\usepackage{microtype}
\usepackage{graphicx}
\usepackage{subfigure}
\usepackage{booktabs} 
\usepackage[utf8]{inputenc} 
\usepackage[T1]{fontenc}    
\usepackage{hyperref}       
\usepackage{url}            
\usepackage{nicefrac}       
\usepackage{microtype}      
\usepackage{wrapfig}
\usepackage{xspace}
\usepackage{algorithm}
\usepackage{algorithmic}
\usepackage{comment}
\usepackage{xspace}
\usepackage{amsmath,  amsthm, thmtools, amsfonts}

\newif\ifappendix
\appendixfalse
\appendixtrue

\newcommand{\refappendix}[1]{\ifappendix
 Appendix~\ref{#1}\xspace
\else
 the supplementary material\xspace
\fi}

\newtheorem{theorem}{Theorem}[section]
\newtheorem{lemma}[theorem]{Lemma}
\newtheorem{definition}[theorem]{Definition}

\newtheorem{corollary}[theorem]{Corollary}

\usepackage{cleveref}
\crefname{theorem}{Theorem}{Theorems}
\Crefname{lemma}{Lemma}{Lemmas}
\Crefname{observation}{Observation}{Observations}
\Crefname{equation}{}{}

\crefname{cl}{Claim}{Claims}

\crefname{obs}{Observation}{Observations}

\crefname{conj}{Conjecture}{Conjectures}

\crefname{lemma}{Lemma}{Lemmas}

\crefname{rem}{Remark}{Remarks}

\newcommand{\cM}{\mathcal{M}}
\newcommand{\opt}{\textrm{Opt} }
\newcommand{\alg}{\textrm{Alg} }
\newcommand{\dist}{\textrm{Dist} }

\newcommand{\mediankanonymization}{smooth-$k$-anonymization\xspace}
\newcommand{\mediankanonymity}{smooth-$k$-anonymity\xspace}
\newcommand{\mediankanonymous}{smooth-$k$-anonymous\xspace}

\usepackage{hyperref}

\title{Smooth Anonymity for Sparse Graphs}

\author{
 Hossein Esfandiari\\
 Google Research
\and
 Alessandro Epasto,\\
 Google Research
\and
 Vahab Mirrokni,  \\
 Google Research
\and
 Andres Munoz Medina,\\
 Google Research
}

\begin{document}
\sloppy

\maketitle

\begin{abstract}
When working with user data providing well-defined privacy guarantees is paramount. In this work, we aim to manipulate and share an entire sparse dataset with a third party privately. In fact, differential privacy has emerged as the gold standard of privacy, however, when it comes to sharing sparse datasets, e.g. sparse networks, as one of our main results, we prove that \emph{any} differentially private mechanism that maintains a reasonable similarity with the initial dataset is doomed to have a very weak privacy guarantee. In such situations, we need to look into other privacy notions such as $k$-anonymity. In this work, we consider a variation of $k$-anonymity, which we call smooth-$k$-anonymity, and design simple large-scale algorithms that efficiently provide smooth-$k$-anonymity. We further perform an empirical evaluation to back our theoretical guarantees and show that our algorithm improves the performance in downstream machine learning tasks on anonymized data.
\end{abstract}

\section{Introduction}
When working with user data, maintaining user privacy is absolutely essential. 
In this work we study a situation where we intend to share an entire (manipulated) dataset, without violating user privacy. Then the dataset might be used by the public for several different purposes. Hence, to measure the accuracy regardless of the downstream task, we use a general purpose metric to measure the similarity of the initial dataset with the shared dataset. To measure the privacy there are a large body of work that attempt to provide formal privacy measures. At a high level, there are two distinct approaches to quantifying privacy, {\em differential privacy} and {\em $k$-anonymity}.

Differential privacy is a property of a data processing algorithm and it ensures that small changes in input (typically the presence or absence of any individual user) lead to minimal changes in the output. All differentially private algorithms are randomized, and the uncertainty introduced by the randomization provides a layer of protection. 
On the other hand, $k$-anonymity is a property of the dataset. To make a dataset $k$-anonymous one either generalizes or removes data that is identifiable, so that in the final dataset any information is shared by at least $k$ distinct users. Both approaches have their own pros and cons, which we briefly discuss next. (We defer the formal definitions to Section \ref{sec:setup}.) 

The main advantage of differential privacy is that the output of a differentially private algorithm remains such even in the face of arbitrary post-processing by an adversary armed with additional side information about the users. This is one reason why it has emerged as the gold standard of privacy.
However, as we share more information the differential privacy measure gets weaker. In this work we prove that sharing sparse binary matrices with differential privacy guarantees is infeasible (See Theorem~\ref{thm:hardness:general}). Roughly speaking, we prove that any differentially private algorithm either provides a very weak privacy guarantee, or significantly changes the dataset, destroying the underlying signal. 

On the other hand, $k$-anonymity \cite{sweeney2002k} is a popular pre-processing technique that can be used to provide some level of privacy. While $k$-anonymity can be vulnerable to certain attacks \cite{Ganta2008attacks}, it still provides meaningful guarantees when adversaries have limited access to side information \cite{BassilyGKS13coupled}. Moreover, in cases where a data analyst cannot withstand noise, it still represents a formal way to give privacy protections. 
Making a dataset $k$-anonymous while best preserving utility is an  NP-hard problem~\cite{aggarwal2005approximation}. Current approximation algorithms offer the guarantee of removing at most $O(\log(k))$ times more elements than that of an optimal solution, however, such a bound is vacuous when the optimal solution has to remove a constant fraction of the dataset (or  anything smaller than a $1-O\left(\frac{1}{\log\left(k\right)}\right)$ fraction). In those cases the algorithm that just returns a null dataset achieves the same guarantee. 

In this work, we strive to design an approach for sharing a binary matrix, while respecting the privacy of the users. 
In order to do this we study a variant of $k$-anonymity (called \mediankanonymity). Then we provide a polynomial-time approximation algorithm for \mediankanonymity in binary matrices and in theory improve the approximation guarantees of the state of the art results for $k$-anonymization.

In the binary matrix representation, each row represents the data of one user and each column corresponds to a feature, and if the user $u$ has the feature $f$, element $(u,f)$ in the matrix is $1$. This representation captures the following common setups: \\
    {\bf Bipartite Graphs:} The nodes of one side correspond to the users and the nodes of the other side corresponds to the features. If user $u$ has feature $f$, there is an edge between $u$ and $f$.\\
    {\bf User Lists:} We have a collection of lists of users, and each list is associated with a feature. If user $u$ has feature $f$, user $u$ exists in the list associated with $f$.\\
    {\bf Points in a Binary Space:} Each user is associated with a point. The coordinates of the point are equivalent to the respective row in the matrix representation.

\section{Related work}
The problem of anonymizing data is very well studied. One of the first techniques for anonymizing data sets was $k$-anonymity~\cite{sweeney2002k}. This notion was intended for tabular data where each row corresponds to a user and each column corresponds to a particular feature. The authors define $k$-anonymity in terms of quasi-identifiers. That is, columns in the data set that, combined, could single out a user. A $k$-anonymous dataset is one where every user is indistinguishable from $k$-other users with respect to the quasi-identifier set only (that means that the columns not corresponding to quasi-identifiers are not anonymized). Other works have improved upon this definition by enforcing other restrictions such as requiring $l$-diversity~\cite{machanavajjhala2007diversity} or $t$-closeness~\cite{li2007t} for non quasi-identifiers, on top of $k$-anonymity for quasi-identifiers. The choice of quasi-identifiers is crucial since an attacker with just a small amount of information about a user could de-anonymize a dataset~\cite{NarayananS08}.


The majority of work on $k$-anonymity has been focused on finding the optimal $k$-anonymous dataset. That is, one that approximates the original data the best. \cite{MeyersonW04} showed that this task is in fact NP-hard, although it admits a $O(k \log k)$ approximation  with running time exponential in $k$. 
Later on, \cite{aggarwal2005approximation} obtained a polynomial time approximation of $O(k)$ which was improved by \cite{kenig2012practical,park2010approximate} to  a $O(\log k)$ approximation. Several variants of these algorithms including using set cover approximations~\cite{wang2010anonymizing}.
In addition to these algorithms with provable guarantees, other work has provided heuristics for different notions of anonymization. For instance~\cite{lefevre2006mondrian} has defined a heuristic algorithm for $k$-anonymization of quasi-identifiers based on the construction similar to that of kd-trees~\cite{friedman1977algorithm}. Other authors have defined heuristics based on clustering~\cite{byun2007efficient,zheng2018k}. None of those methods have provable guarantees in our context. 

Another related work to ours is that of~\cite{Cheng10}, which defines the notion of $k$-isomorphism in social network graphs. Essentially, a graph is $k$-isomorphic if it can be decomposed into a union of $k$ distinct isomorphic sub-graphs. This notion of anonymity is limited to social network graphs as the goal is to prevent an attacker from identifying a user based on the structure of their neighborhood. 

Another framework for achieving anonymity is differential privacy~\cite{dwork}. Unlike $k$-anonymity, differential privacy provides mathematical guarantees on the amount of information that can be gained by an attacker that observes a differentially private dataset. Differential privacy has been effectively applied for statistics release~\cite{DworkS10} and empirical risk minimization~\cite{ChaudhuriMS11}  among many other scenarios. The vast majority of differential privacy examples require the mechanism to output a summarized version of the data: a statistic or a model in the case of risk minimization. To release a full dataset in a differentially private manner, \cite{KasiviswanathanLNRS11} introduces the notion of local differential privacy. Local differential privacy allows us to release a full dataset while protecting the information of all users.

Methods from differential privacy have also been used for the release of private graph information. For instance, \citet{NissimRS07} shows how to compute the minimum spanning trees and the number of triangles in a graph. \citet{eliavs2020differentially} recently showed how to preserve cuts in graphs with differential privacy.

 \citet{nodeprivacy} introduce the notions of edge and node differential privacy in graph settings and show how to calculate functions over graphs under node differential privacy by capping the number of edges per node. Arguably the work most related to this paper is that of \citet{NguyenIR16}. The authors propose edge-differential privacy in order to release an anonymous graph. Similar to our results in Corollary~\ref{cor:high_eps}, the authors show that a value of $\epsilon$ in $\Omega(\log n)$ is needed in order to achieve non-trivial utility guarantees, where $n$ is the number of nodes in the graph. 

We observe that some work has been devoted to combining differential privacy with $k$-anonymity guarantees. ~\citet{li2011provably} show that enforcing $k$-anonymity in certain data-oblivious ways on a sub-sampled dataset, is sufficient to show differential privacy guarantees.

Our algorithmic techniques are related to the lower bounded facility location problem~\cite{ahmadian2012improved,guha2000hierarchical,svitkina2010lower}. This problem has been first introduced and studied independently by ~\citet{karget2000building} and~\citet{guha2000hierarchical}. Lower bounded clustering problems have been motivated by privacy purposes~\cite{motwani2008anonymizing}. They both provide bicriteria approximation algorithms for this problem. Later, ~\citet{svitkina2010lower} provided a $448$-approximation algorithm for this problem. This is the first constant approximation algorithm for this problem.~\citet{ahmadian2012improved} improved Svitkina's result and give an $82.6$ approximation algorithm for this problem. To the best of our knowledge the latter is the best approximation algorithm for the lower bounded facility location problem. 

\section{Setup}
\label{sec:setup}
Given the equivalence of binary matrices and bipartite graphs, for ease of notation we mostly use graph theoretical terminology to describe our work. We assume we are given a bipartite graph, where one set of nodes corresponds to users and another set of nodes corresponds to features. This is a common modeling step, for instance in location analysis applications the features may represent places visited; in social network modeling, the features may represent interests shared by different users; and so on. 

Let $U = \{u_1, \ldots, u_n\}$ denote a set of users and $F= \{f_1,\ldots, f_m\}$ a set of features. Throughout the paper we use $n$ and $m$ as the $|U|$ and $|F|$, respectively. The edge set $E$  of the graph is defined as follows, given $u \in U$ and $f\in F$, we say $e = (u, f) \in E$ if user $u$ is associated with item $f$. We denote this graph by $G = (U \cup F, E)$. Let $\mathbb{G}$ denote the space of all bipartite graphs over $U\cup F$, a mechanism $\cM \colon \mathbb{G} \to \mathbb{G}$ is a (possibly randomized) function that maps $G = (U \cup F, E)$ to another graph $G' = (U \cup F, E')$ with the same set of nodes but with possibly different edges. Given two sets $A$ and $B$ we denote their symmetric difference by $A \oplus  B$.

We now introduce the different notions of privacy we will be using throughout the paper. 

\begin{definition}{Edge differential privacy.}
We say a randomized mechanism $\cM$ preserves $\epsilon$-edge differential privacy if for any two graphs $G = (U \cup F, E)$ and $G' = (U \cup F, E')$ such that $|E \oplus E'| = 1$ the following holds for all $A \subset \mathbb{G}$:
\begin{equation*}
    P(\cM(G) \in A) \leq e^{\epsilon} P(\cM(G') \in A),
\end{equation*}
\end{definition}
Edge differential privacy implies that the output of a mechanism does not change too much if a single edge of the input graph is changed. Thus an adversary that observes the output of $\cM(G)$ may not be able to infer if a single edge was present or not in the graph. 

However, if a user has a high degree in $G$, then the output of an edge-differentially private algorithm may still leak information about the presence or absence of that user in the graph. This leads to a definition of {\em node-differential privacy} which we detail below. 

\begin{definition}
We say that two graphs $G=(U \cup F, E), G'= (U \cup F, E') \in \mathbb{G}$ are node neighboring if there exists $u \in U$ such that $|E \oplus E'| = |\{e \in E \colon e = (u, f) \ \text{for }\ f \in F\} \oplus \{e \in E' \colon e = (u, f) \ \text{for }\ f \in F\} $
\end{definition}

That is, two graphs are node neighboring if one can be obtained from the other by replacing all the edges of a single user. 
\begin{definition}[Node differential privacy]
We say a mechanism $\cM$ preserves node differential privacy if for any two node neighboring graphs $G$ and $G'$ and for all $A \subset \mathbb{G}$
\begin{equation*}
    P(\cM(G) \in A) \leq e^{\epsilon} P(\cM(G') \in A).
  \end{equation*}
\end{definition}
Under this notion of anonymity, it is very unlikely for an adversary to identify a user in a particular dataset. While the above notions of differential privacy provide quantifiable protection against an attacker, as we will see, they also require adding a non-trivial amount of noise. 

For this reason we revisit an older notion of privacy: $k$-anonymity.  While the original definition of $k$-anonymity~\cite{sweeney2002k} requires defining quasi-identifiers, in this work we assume that every feature can be used as a quasi-identifier.  

We first introduce some notation. We will consider graphs in $\mathbb{G}$ with fixed node sets $U \cup F$, and varying edge sets. Let $G = (U \cup F, E) \in \mathbb{G}$ be one such graph, notice that the graph is identified by $E$. For a given edge set $E$, let $F_u(E) = \{f \in F \colon (u, f) \in E\}$ be the items associated with $u$ in the set edge set $E$. Notice we can then partition users into equivalence classes. Formally, let 
$$C_u(E) = \{ u' \in U | F_u(E) \equiv F_{u'}(E) \}.$$
%
Now we are ready to formally define $k$-anonymity by suppression.
\begin{definition}[$k$-anonymization and $k$-anonymization by suppression]
\label{def:anonfull}
A mechanism $\cM$ is $k$-anonymous if for any graph $G = (U \cup F, E) \in \mathbb{G}$,  $\cM(G) = (U \cup F, E')$ satisfies:
\begin{enumerate}\itemsep=0in
    \item For every $u\in U$, $|C_u(E')| \geq k$.
\end{enumerate}
The mechanism $\cM$ is $k$-anonymous by suppression if it also satisfies
\begin{enumerate}\itemsep=0in

    \item $E' \subset E$
\end{enumerate}
\end{definition}
That is the set of items associated with each user in the output graph, is the same of that of at least $k$  users. Moreover, in $k$-anonymity with suppression the output set of edges $E'$ needs to be a subset of $E$. Notice that with a $k$-anonymous output an adversary can only distinguish a user up to a set of $k$ different people. 

Finally, we introduce our variant of the above definition. 

\begin{definition}[\mediankanonymity]
\label{def:majorityanon}
A mechanism $\cM$ is \mediankanonymous if for any graph $G = (U \cup F, E) \in \mathbb{G}$,  $\cM(G) = (U \cup F, E')$ satisfies:
\begin{enumerate}\itemsep=0in
    \item For every $u \in U$, $|C_u(E')| \geq k$.
    \item For every $u \in U$, and every $f\in F$, $(u,f) \in E'$ implies $|\{u' \in C_u(E') \colon (u',f) \in E' \}| \ge \nicefrac{|C_u(E')|}{2}$
\end{enumerate}
\end{definition}
This definition is very similar to Definition~\ref{def:anonfull}. The main difference between the definitions is that a \mediankanonymous mechanism is only allowed to add edges to the output if, for each equivalence class of users and each item connected to them, the majority of such edges belong to the original graph. Figure~\ref{fig:depiction} we depict the difference between our \mediankanonymous and k-anonymity with suppression. 

\begin{figure}
\centering
\includegraphics[width=0.25\textwidth,keepaspectratio]{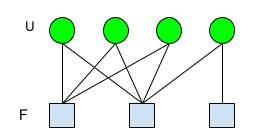}
\includegraphics[width=0.25\textwidth,keepaspectratio]{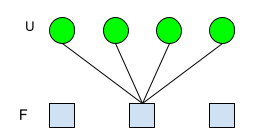}
\includegraphics[width=0.25\textwidth,keepaspectratio]{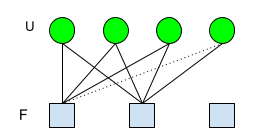}
\caption{Depiction of $k$-anonymity with suppression and \mediankanonymity for $k=4$. (left) Original input graph $G$. (center) $k$-anonymous with suppression graph. Notice the removal of the edges to the first and last feature.  (right) \mediankanonymous graph, we preserve the edges to the first feature and add a new edge to it.} 
\label{fig:depiction}
\end{figure}

We conclude this section by defining the utility measure of a mechanism. In order for a mechanism to be useful it should preserve as much as possible of the graph structure. In this paper we measure this by the Jaccard similarity of two graphs.
\begin{definition}
  Given two graphs $G = (U \cup F, E), G'  = (U \cup F, E')$ we denote the Jaccard similarity of them by
 $J(G, G') := \frac{|E \cap E'|}{|E \cup E'|}$.
  \end{definition}

\section{Comparison of privacy notions}\label{sec:dp-comparison}

Here we introduce the new algorithmic problem of finding the best smooth-$k$-anonymization of a graph. For this reason, in this section, we provide some comparison between \mediankanonymity and alternative privacy notions that can be used for data release. 

\subsection{Comparison with differential privacy}

\textbf{Node differential privacy}
As we have briefly discussed, node-differential privacy provides the best theoretical guarantees for privacy protection. In the specifics of our setup, node-differential privacy is equivalent to the so-called \emph{local} differential privacy \cite{KasiviswanathanLNRS11}, where every user is acting separately without coordination from some global authority. 
  
Let $G = (U \cup F, E)$, borrowing from local differential privacy ideas, one way of achieving node differential privacy is by releasing $\mathcal{M}(G) = (U \cup F, E')$ built according to Algorithm~\ref{alg:randomized_response}. The algorithm is parameterized by a randomized response probability $p$. It is not hard to show that in order to achieve $\epsilon$-node differential privacy $p = \frac{2}{1 + e^{\frac{\epsilon}{|F|}}}$. Notice that this value converges to $1$ exponentially fast as a function of the size of the feature set $F$. That is, even for relatively small graphs, in order to achieve any meaningful privacy guarantee, the probability of returning a completely random graph is very close to $1$. For this reason, this notion is not amenable to be used with good utility in our setting. 
 
  \begin{algorithm}[t]
    \caption{ Randomized response}
    \label{alg:randomized_response}
    \begin{algorithmic}
      \STATE {\bf Input:} $G = (U \cup F, E)$, randomized response prob. $p$
      \STATE {\bf  Output:} anonymized graph $G' = (U \cup F, E')$.
      \STATE {\bf for} $u \in U , f \in F$
      \STATE $\quad$ Sample $Y \sim \textrm{Bernoulli}(p)$
      \STATE $\quad$ {\bf if}  $Y = 1$ {\bf then} $(u, f) \in E'$ with probability $1/2$.
      \STATE $\quad$ {\bf else} {\bf if } $(u, f) \in E$ {\bf then}  $(u, f) \in E'$
      \end{algorithmic}
    \end{algorithm}
\textbf{Edge differential privacy}
Algorithm~\ref{alg:randomized_response} can also be used to define a mechanism $\mathcal{M}$ that is edge differential privacy. In that case one can achieve $\epsilon$-edge differential privacy by setting $p = \frac{2}{1 + e^\epsilon}$  (see \refappendix{app:dp}).

In this section fist we consider Algorithm~\ref{alg:randomized_response} as a natural way to provide differential privacy
and upper bound the Jaccard similarity of the input and the output graphs of this algorithm. Later in this section we show that a similar bound holds for all differential privacy algorithms.

\begin{figure}[t!]
\centering
\includegraphics[width=0.5\textwidth,keepaspectratio]{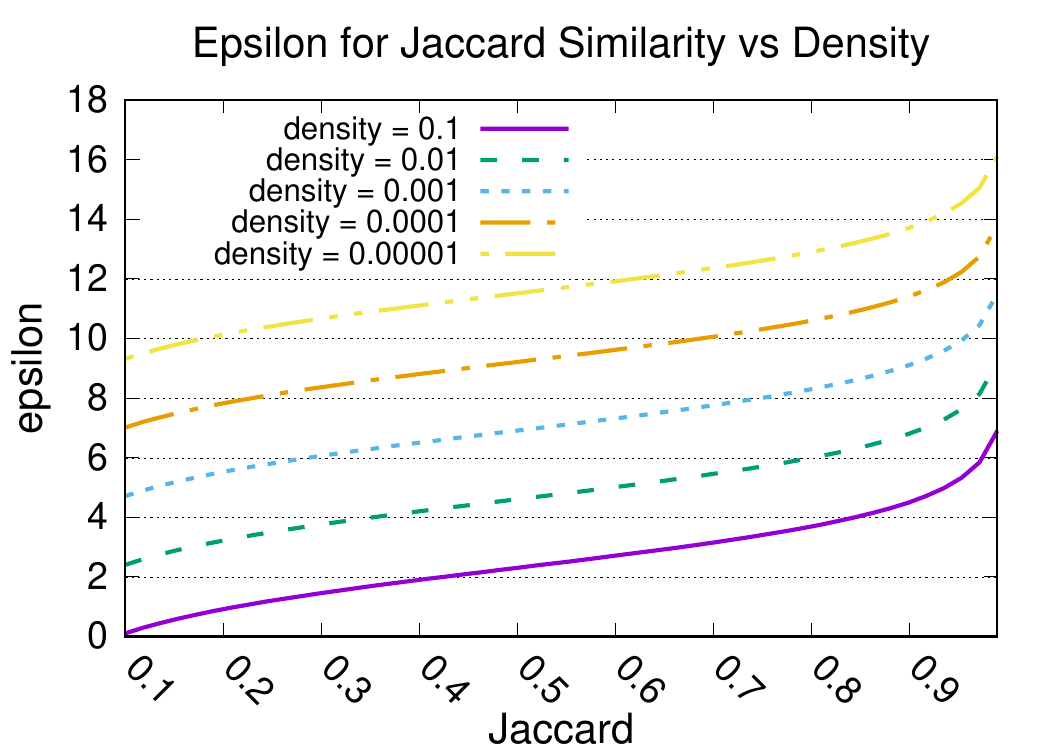}
\caption{The $\epsilon$ necessary for a given Jaccard, as a function of density \label{fig:result-eps-for-jaccard}}.
\end{figure}

We now upper bound the similarity of the output of the algorithm with $\epsilon$-edge differential privacy. The bound depends on the density $\lambda$ of the graph $G$ defined as:
$\lambda(G) = \frac{|E|}{|U||F|}.$ The proof of this theorem is available in~\refappendix{prop:upperbound-eps}.

\begin{theorem}\label{prop:upperbound-eps}
Let $G=(U \cup F, E) \in \mathbb{G}$ be a graph, $\delta > 0$. Let $\cM$ be a mechanism that generates a graph according to Algorithm~\ref{alg:randomized_response} with $p = \frac{2}{1 + e^\epsilon}$.
Let $Q = |U||F|$ and $C(\delta) = \frac{\log(2/\delta)}{2}$. If $\frac{p}{4}\geq \sqrt{\frac{C(\delta)}{Q}}$, then with probability at least $1 - \delta$:
$$J(\cM(G), G) \leq \frac{1 - \frac{1}{e^{\epsilon} + 1}}{1 + \frac{1}{e^{\epsilon} + 1}\frac{1 - \lambda}{\lambda}} + 2\sqrt{\frac{C(\delta)}{Q}},$$
where $\lambda := \lambda(G)$.
\end{theorem}

Using Theorem~\ref{prop:upperbound-eps}, we can plot in Figure~\ref{fig:result-eps-for-jaccard} a lower bound for the $\epsilon$ needed to achieve a certain level of Jaccard similarity utility, given a density factor $\lambda$. Notice that for reasonably sparse datasets, say graphs with density around $1/10{,}000$, one needs an $\epsilon$ higher than $10$ to obtain a Jaccard Similarity of more than $50\%$. This result will be confirmed in Section~\ref{sec:exp} in our empirical analysis.

The previous Theorem allows us to derive the following.
\begin{corollary}
\label{cor:high_eps}
If the average degree is $O(m^{0.99})$, for any $\epsilon \in o(\log m)$, $\epsilon$-DP gives a solution with Jaccard similarity of $o(1)$.
\end{corollary}

This means that, when the average user-degree is poly-logarithmic (or even $m^{0.99}$) we need $\epsilon \in \Omega (\log m)$ to achieve a constant Jaccard similarity with $\epsilon$-edge differential privacy. It is well-known that many real-world  datasets are sparse. For instance in the context of real-world networks, classical theoretical models~\cite{albert2002statistical} as well as empirical studies~\cite{backstrom2012four,leskovec2007graph} postulate constant average degree or degrees growing slower than the size of the graph.    

The above results show that using randomized response, any differentially private approximation to a graph with high utility requires an exceptionally large value of $\epsilon$, thus rendering void any privacy guarantees. Later in Theorem~\ref{thm:hardness:general} we show a similar bound on $\epsilon$ for all differentially private algorithms. In fact, for such a high $\epsilon$, the algorithm likely maintains the graph $G$ unmodified thus exposing users to re-identification risks. For this reason we believe $k$-anonymity might in fact provide better protection in practice, especially in scenarios where our goal is to produce a private version of the input graph with high utility.

Unfortunately, however, all of the previous work~\cite{aggarwal2005approximation,kenig2012practical} provide non-trivial guarantees on the quality of a $k$-anonymous mechanism only when  $J(E,E_{\opt}) \geq 1-O(\frac 1 {\log k})$, in other words when very few edges need to be removed, as the similarity between the graph and the optimal $k$-anonymous graph is very high.  By contrast, we show in Section~\ref{sec:technical} one can achieve a constant approximation to the optimal \mediankanonymous solution in less restrictive scenarios. 

\textbf{Hardness of Differential Privacy}\label{sec:hardness:general}
So far in this section we analyzed randomized response and show that this mechanism requires an unacceptably large $\epsilon$. One may ask if there is any other $\epsilon$-differential privacy mechanism with a small $\epsilon$ that guarantees the output to be similar to the input. The following theorem rules out the existence of such a mechanism. 
 
\begin{theorem}\label{thm:hardness:general}
Let $\mathcal{M}$ be an arbitrary mechanism that satisfies $\epsilon$-edge differential privacy. Let $\alpha$ be a parameter such that for any input graph $G = (U \cup F, E)$, we have $\alpha \leq E[J(\mathcal{M}(G),G)]$. We have
$\epsilon \in \Omega\big(\log(\alpha^2 nm)\big).$
\end{theorem}

\begin{proof}
To prove this first we define a policy $\overline{\mathcal{M}}(G)$ based on $\mathcal{M}(G)$, such that $\overline{\mathcal{M}}(G)$ is 
\begin{itemize}
    \item an $\epsilon$-differentially private mechanism,
    \item $E[J(G,\overline{\mathcal{M}}(G))]\geq \frac{\alpha}{2}$, when $|G|\geq l=\lfloor(nm)^{0.9}\rfloor$, and
    \item $|\overline{\mathcal{M}}(G)|\leq\frac{2(l+1)}{\alpha}$.
\end{itemize}
The third property bounds the range of $|\overline{\mathcal{M}}(G)|$ and allows us to analyze $\overline{\mathcal{M}}(G)$ and bound $\epsilon$.
We define policy $\overline{\mathcal{M}}(G)$ based on $\mathcal{M}(G)$ as follows:
\begin{itemize}
    \item If $|\mathcal{M}(G)|>\frac{2(l+1)}{\alpha}$ then $\overline{\mathcal{M}}(G)$ is set to empty graph,
    \item  otherwise $\overline{\mathcal{M}}(G)=\mathcal{M}(G)$.
\end{itemize}
Note that $\overline{\mathcal{M}}(G)$ can be exactly calculated given $\mathcal{M}(G)$, hence $\overline{\mathcal{M}}(G)$ is an $\epsilon$-differentially private policy as well. 
Moreover, note that when $|\mathcal{M}(G)|>\frac{2(l+1)}{\alpha}$ we have 
\begin{align} \label{eq:hardness:eq1}
J(G,\mathcal{M}(G))=\frac{|G\cap \mathcal{M}(G)|}{|G\cup \mathcal{M}(G)|}\leq \frac{|G\cap \mathcal{M}(G)|}{|\mathcal{M}(G)|} < \frac{\alpha}{2}\frac{|G|}{l+1}.
\end{align}
Hence, for a graph $G$ with $|G|\leq l+1$ we have
\begin{align*}
    E[J(G,\overline{\mathcal{M}}(G))] =
    &E[J(G,{\mathcal{M}}(G))] - E[J(G,{\mathcal{M}}(G)) - J(G,\overline{\mathcal{M}}(G))]\geq\\
    &\alpha - E[J(G,{\mathcal{M}}(G)) - J(G,\overline{\mathcal{M}}(G))]\geq & \text{\it (By def. )}\\
    & \alpha - E[\max(0, \frac{\alpha}{2}\frac{|G|}{l+1})]\geq &\text{ \it (By Ineq.~\ref{eq:hardness:eq1})}\\
    & \frac{\alpha}{2}. & \text{\it (By $|G|\leq l+1$)}
\end{align*}

Consider the following two equivalent random processes to construct random graphs $G = (U \cup F, E)$ and $G' = (U \cup F, E')$.

\begin{itemize}
    \item Select $l$ pairs of nodes from $U\times F$ uniformly at random without replacement. Add an edge between each selected pair in both $D$ and $D’$. Select one other pair of nodes from $U\times F$ uniformly at random without replacement, denote it as $(u,f)$, and add an edge between $u$ and $f$ in $D'$.
    
    \item Select $l+1$ pairs of nodes from $U\times F$ uniformly at random without replacement. Add an edge between each selected pair in both $D$ and $D’$. Select one of the edges in $D$ uniformly at random, denote it as $(u,f)$, and remove it from $D$.
\end{itemize}
Note that, $G$ is a graph chosen uniformly at random from all graphs on $U\times F$ with $l$ edges, and $G'$ is a graph chosen uniformly at random from all graphs on $U\times F$ with $l+1$ edges. Moreover, $(u,f)$ is both an edge selected uniformly at random from the edges inside $G'$ and it is an edge selected uniformly at random from the edges that are not in $G$.

Recall that, by definition, we have 
\begin{align*}
\frac{\alpha}{2} &\leq E\big[J(\overline{\mathcal{M}}(G'),G')\big]= E\Big[\frac{|\overline{\mathcal{M}}(G')\cap G'|}{|\overline{\mathcal{M}}(G')\cup G'|}\Big] 
\leq  E\Big[\frac{|\overline{\mathcal{M}}(G')\cap G'|}{| G'|}\Big] =\frac{E\big[|\overline{\mathcal{M}}(G')\cap G'|\big]}{| G'|}.
\end{align*}
Note that, if we select one of the edges of $G'$ uniformly at random, it exists in $\overline{\mathcal{M}}(G')$ with probability at least $\frac{E\big[|\overline{\mathcal{M}}(G')\cap G'|\big]}{| G'|} \geq \frac{\alpha}{2}$. Hence, we have $(u,f)\in \overline{\mathcal{M}}(G')$ with probability at least $\frac{\alpha}{2}$. Let $S$ be the set of all possible outputs of $\overline{\mathcal{M}}(G')$ where the $(u,f)\in \overline{\mathcal{M}}(G')$. By the definition of differential privacy we have
\begin{align*}
    \textbf{Pr}\Big(\overline{\mathcal{M}}(G')\in S\Big) \leq e^{\epsilon} \textbf{Pr}\Big(\overline{\mathcal{M}}(G)\in S\Big),
\end{align*}
Which means
\begin{align*}
\textbf{Pr}\Big(\overline{\mathcal{M}}(G)\in S\Big) \geq 
    e^{-\epsilon}\textbf{Pr}\Big(\overline{\mathcal{M}}(G')\in S\Big)   \geq \frac{\alpha e^{-\epsilon}}{2}.
\end{align*}
This means that $(u,f)\in \overline{\mathcal{M}}(G)$ with probability at least $\frac{\alpha e^{-\epsilon}}{2}$.
Recall that, by definition $(u,f)$ is an edge chosen uniformly at random from the edges that do not exist in $G$. Hence, if we select one of the edges that do not exist in $G$, it exists in $\overline{\mathcal{M}}(G)$ with probability at least $\frac{\alpha e^{-\epsilon}}{2}$.

On the other hand, similar to $G'$, if we select one of the edges of $G$ uniformly at random, it exists in $\overline{\mathcal{M}}(G)$ with probability at least $\frac{\alpha}{2}$. Hence, we have $$E[|\overline{\mathcal{M}}(G)|] \geq (nm-l)\frac{\alpha}{2} e^{-\epsilon} + l\frac{\alpha}{2} \geq \frac{nm\alpha e^{-\epsilon}}{2}.  $$ 
Recall that by construction we have $|\overline{\mathcal{M}}(G)|\leq\frac{2(l+1)}{\alpha}$. This together with the above inequality gives us $\frac{nm\alpha e^{-\epsilon}}{2}\leq\frac{2(l+1)}{\alpha}$. This implies $\epsilon\geq \log\frac{\alpha^2nm }{4(l+1)} \in \Omega\big(\log(\alpha^2 nm)\big) $, as claimed.
\end{proof}

\subsection{Comparison with $k$-anonymity by suppression}
In this section we compare $k$-anonymity by suppression with \mediankanonymity. 
First, in terms of privacy, we notice that both \mediankanonymity and $k$-anonymity by suppression guarantee that every user in the output is indistinguishable from at least $k$-users. Moreover, observe that since \mediankanonymity is allowed to add edges to the output graph, an attacker would not be certain whether an edge was in the original graph or not. 

We now show formally that the optimum solution of \mediankanonymity may preserve a significantly larger fraction of the input data than regular $k$-anonymity. To show this separation rigorously, we adopt the bipartite stochastic block model (SBM), which is commonly used in modeling applications, for instance in clustering and community detection~\cite{abbe2017community}. 

\textbf{Bipartite Stochastic Block Model.} 
\label{sec:stochastic}
To define the bipartite SBM, consider the following random process. We have two sets of $n$ vertices, and each set is further decomposed into $r$ blocks of size $s$, where $r\cdot s = n$. Each block in the first part corresponds to one block of the second part. There is an edge between each pair of vertices in two corresponding blocks independently with probability $q$, and between every other pair of vertices with probability $p$. We let $\alpha = q s$ denote the expected number of edges that one node has to its corresponding block, a.k.a. \emph{internal edges}. 
We let $\beta = p (n-s)$ denote the expected number of edges that one node has to vertices other than its corresponding block, a.k.a. \emph{external edges}. We refer to this as the stochastic block model with parameters $r,s,\alpha,\beta$. 

The result that we provide in this section is of particular interest when the blocks are not very sparse, i.e., $\alpha \in \omega(\frac {\log n}{\log 1/q})$ and $\alpha \in \Omega(\beta + s)$.

The next theorem upper bounds the number of edges in a (non-smooth) $k$-anonymous subgraph of a graph generated by the stochastic block model by $O(n \frac {\log n}{\log 1/q})$. Therefore, since $\alpha \in \omega(\frac {\log n}{\log 1/q})$, the fraction of remaining edges tends to zero. The proof of this theorem is presented in \refappendix{app:thm:sbm}.
\begin{theorem}\label{thm:sbm}
Let $G$ be a graph generated by the stochastic block model with parameters $r,s,\alpha,\beta$. Let $k \geq \frac{2\log n}{\log  1 / q} $. With probability $99\%$, any $k$-anonymous subgraph of $G$ contains at most $\frac{2\log n+10}{\log  1 / q} n  \in \tilde{O}(n)$ edges. 
\end{theorem}
This allows us to show a gap with \mediankanonymization. In fact, a natural solution that puts the vertices of each block in a cluster leads to a solution for \mediankanonymity that in expectation keeps $\alpha n$ edges, adds $(s-\alpha)n$ edges, and removes $\beta n$ edges. Since $\alpha \in \Omega(\beta + s)$, the number of remaining edges $\alpha n$ is not less than a constant factor of the changed edges. This result concerns the optimum solution, but in the next section we provide an algorithm for computing \mediankanonymization of a graph.
\section{Algorithms and Analysis}
\label{sec:technical}
In this section we develop algorithms that find a \mediankanonymization of $G$.
We say an algorithm $alg$ is $\alpha$-approximation if $J(E,E_{alg})/J(E,E_{\opt}) \geq \alpha$, where $E_{alg}$ is the output of $alg$, $E_{\opt}$ is the optimal solution, and $J(\cdot,\cdot)$ is the Jaccard similarity function.
%
%
Our main contribution is captured by the following theorem. 

\begin{theorem}
Assume $J(E,E_{\opt}) \geq 0.75$. There exists an algorithm that finds a constant approximate \mediankanonymization of $G$ in polynomial time.
\end{theorem}

At a high level, our algorithm decomposes the users into clusters, each of size at least $k$. Then in each cluster $c$, for each item $f$, if the majority of the vertices in  $c$ have an edge to $f$, it adds edges to $f$ from all nodes in $c$; otherwise it removes the edges to $f$ from all nodes in $c$. 

\subsection{Preliminaries}
Let us start with some preliminary notions and lemmas. We will abuse notation slightly, and for a user $u$ and item $f$, say $u \in f$ if there is a $(u,f)$ edge in the graph. 

Note that we can represent each user $u$ with a point in a $m$ dimensional space. For an item $f_i$, we set the $i$-th dimension of $u$'s representation to $1$ if $u\in f_i$ and to $0$ otherwise. In this space we define the distance of two points, $u$ and $v$ to be the number of positions where $u$ and $v$ differ ({\em i.e.} their Hamming distance). 

Let $\Delta_u^{\alg}$ be the number of positions that the algorithm $\alg$ changes in the binary vector corresponding to the user $u$. Intuitively, $\sum_u  \Delta_u^{\opt}$ should be related to $J(E, E_{\opt})$. The following lemma formalizes this intuition.


\begin{lemma}\label{lm:Jaccard2Count}
Assume $J(E,E_{\opt}) \geq 1-\phi$. We have $\sum_{u} \Delta_u^{\opt} \leq \frac{2\phi}{1-\phi} |E|$.
\end{lemma}
All proofs from this section are in~\refappendix{app:sec:technical}.

To complement Lemma \ref{lm:Jaccard2Count}, the following lemma lower bounds $J(E,E_{\alg})$ given an upper bound on $\sum_{u} \Delta_u^{\alg}$.

\begin{lemma}\label{lm:Count2Jaccard}
Assume $\sum_{u} \Delta_u^{\alg} \leq \phi' |E|$. We have $J(E,E_{\alg}) \geq 1-\frac{\phi'}{2}$.
\end{lemma}

\subsection{Initial algorithm}
We now provide an approximation algorithm using a reduction to \emph{lower-bounded $r$-median}\footnote{We use $r$-median instead of $k$-median to avoid the confusion with the parameter $k$ in $k$-anonymity. We will use $r$ as the upper bound on the number of clusters and $k$ as the lower bound for the size of the clusters.}. In the next subsection we improve it using a slightly more complicated algorithm. 

In the lower-bounded $r$-median problem we are asked to select at most $r$ centers from $n$ points and assign each point to one center such that (i) the number of points assigned to each center is at least $k$, (ii) the total distance of the points from their assigned centers is minimized. We refer to each set of the points that are assigned to the same center as a cluster. In this paper we let $r=n/k$, which means that the algorithm may use as many centers as it needs, however, it must assign at least $k$ points to each center\footnote{This means there are at most $n/k$ centers.}. Here we use a $82.6$ approximation algorithm for lower-bounded $r$-median~\cite{ahmadian2012improved}, which is the best known result to the best of our knowledge.  We refer to this algorithm as $\alg_1$.
\begin{enumerate}
    \itemsep 0em
    \item Embed each user in $\mathbb{R}^m$ as described at the beginning of this section.
    \item Approximately solve the lower-bounded $r$-median on the points (for $r=\nicefrac{n}{k}$).
    \item For each cluster $c$, for each item $f$, if most vertices in $c$ have an edge to $f$, add all edges from nodes in $c$ to $f$, otherwise  remove all edges from nodes in $c$ to $f$. 
\end{enumerate}
Note that by definition of lower-bounded $r$-median, each cluster contains at least $k$ points. Moreover, the data that we output for users that belong to the same cluster are the same. Hence, the output satisfies the anonymity part of the \mediankanonymity condition. Moreover, the output satisfies the majority part of the \mediankanonymity assumptions. Next we bound $J(E,E_{\alg_1})$ assuming $J(E,E_{\opt})\geq 1-\phi$.

For analysis sake, we introduce the \emph{relaxed lower-bounded $r$-median} problem in which we are allowed to select any possible discrete point in the space as a center (as opposed to being restricted to select centers only from the points that appear in the input). Note that, if we take a solution to relaxed lower-bounded $r$-median and move each center to its closest point (that appears in the input), by triangle inequality the cost of the solution increases by at most a factor $2$. Therefore the cost of lower-bounded $r$-median is at most twice that of relaxed lower-bounded $r$-median.

By Lemma \ref{lm:Jaccard2Count} we have $\sum_{u} \Delta_u^{\opt} \leq \frac{2\phi}{1-\phi} |E|$. We now prove there exists a solution to relaxed lower-bounded $r$-median with cost at most $\frac{2\phi}{1-\phi} |E|$. Take an optimal anonymous solution and consider the equivalence classes of nodes with the same neighborhood. This induces a clustering of the nodes with clusters of size at least $k$. Now, observe that the total number of entries changed is equal to the sum of distances from the output neighborhood (of each class) and the original nodes. So this shows that there exists a clustering with sizes at least $k$ with total cost $\sum_{u} \Delta_u^{\opt}$.

Therefore, there exists a solution to lower-bounded $r$-median with cost at most $\frac{4\phi}{1-\phi} |E|$. Note that we are using an $82.6$-approximation algorithm to find lower-bounded $r$-median. Hence, the total cost of our solution is at most $\frac{330.4\phi}{1-\phi} |E|$. The last line of the algorithm does not increase the total cost (since it selects the best center for each cluster). Hence we have $\sum_{u} \Delta_u^{\alg_1 } \leq \frac{330.4\phi}{1-\phi} |E|$. By applying this to Lemma \ref{lm:Count2Jaccard} we have $J(E,E_{\alg_1}) \geq 1- \frac{165.2\phi}{1-\phi}$. This is a positive constant for any $\phi \leq 0.006$. This implies the following theorem.

\begin{theorem}
Assume $J(E,E_{\opt}) \geq 0.994$. There exists an algorithm that finds a constant approximation  \mediankanonymization  of $G$ in polynomial time.
\end{theorem}
Of course, having $J(E,E_{\opt}) \geq 0.994$ is a very strong assumption. Next we substantially relax this requirement.

\subsection{Improved algorithm}\label{alg:improved}
 To prove a better algorithm we will use the $1.488$ approximation algorithm for the metric facility location problem~\cite{li20131} as a subroutine. In the metric facility location problem we are given a set of points and a set of facilities in a metric space, with an opening cost for each facility. The objective is to select a set of facilities and assign each point to a facility such that the total cost of the selected facilities plus the total distance of the points from their assigned facilities is minimized. Again here, we refer to the set of points assigned to each facility as a cluster. 

Below is our algorithm $\alg_2$. This algorithm depends on a parameter $\alpha$ that we set later.
\begin{enumerate}
    \itemsep 0em
    \item Embed each user in $\mathbb{R}^m$ as before.
    \item For each user $u_i$ define a facility with the same coordinates and opening cost $\frac{2\alpha}{1-\alpha} \sum_{u'\in U^k_i} \dist(u',u_i)$, where $U^k_i$ is the set of $k$ closest points to $i$.\label{line2}
    \item Approximately solve the facility location instance. \label{line3}
    \item Iteratively, remove each cluster with fewer than $\alpha k$ points and assign its points to their second closest facility. \label{line4.1}
    \item Arbitrarily merge clusters with size less than $k$ to reach size $k$, but do not let the clusters grow larger than $2k$.~\footnote{If needed, break a large cluster into some clusters of size at least $\alpha k$, so that the total size of small clusters is more than $k$. This modification does not change the proof.} \label{line4.2}
    \item For each cluster $c$, for each item $f$, if most vertices in $c$ have an edge to $f$, add all edges from nodes in $c$ to $f$, otherwise  remove all edges from nodes in $c$ to $f$. 
\end{enumerate}
%
\begin{theorem}\label{thm:main:improved}
Assume $J(E,E_{\opt}) \geq 0.75$. Algorithm $\alg_2$ run with an $\alpha \in O(1)$ finds a constant approximation to \mediankanonymization of $E$ in polynomial time.
\end{theorem}
%
\begin{proof}
Svitkina~\cite{svitkina2010lower} showed that Lines \ref{line2} to \ref{line4.2} give solution in which (i) the size of each cluster is at least $\alpha k$, and (ii) the cost of the solution is at most $1.488 \cdot\frac{1+\alpha}{1-\alpha}$ times that of lower-bounded $r$-median.
Similar to our previous algorithm, since the size of each cluster is at least $k$, the last line of the algorithm guarantees  \mediankanonymity. Next we bound $J(E,E_{\alg_2})$ assuming $J(E,E_{\opt})\leq 1-\phi$.
We refer to the first four lines of algorithm $\alg_2$ as $\alg'_2$.
Similar to the previous subsection we know that there exists a solution to lower-bounded $r$-median with cost at most $\frac{4\phi}{1-\phi} |E|$. Therefore the cost of the solution to the facility location problem is at most 
$$5.952 \cdot\frac{1+\alpha}{1-\alpha} \cdot \frac{\phi}{1-\phi} |E|.$$ 
Again similar to the previous subsection and by applying Lemma \ref{lm:Count2Jaccard} we have 
$$J(E,E_{\alg'_2}) \geq 1- 2.976 \cdot\frac{1+\alpha}{1-\alpha} \cdot \frac{\phi}{1-\phi}.$$ 
Later we show that Line \ref{line4.2} decreases the Jaccard similarity by at most a factor $\frac{\alpha^2}{8}$. Hence we have $$J(E,E_{\alg_2}) \geq\frac{\alpha^2}{8}\Big(1- 2.976 \cdot\frac{1+\alpha}{1-\alpha} \cdot \frac{\phi}{1-\phi}\Big),$$ 
which is a positive constant for $\phi \leq 0.25$ and $\alpha=0.004$. This implies Theorem~\ref{thm:main:improved}.

New we show that Line \ref{line4.2} decreases the Jaccard similarity by a factor of at least $\frac{\alpha^2}{8}$. To prove this, we use the probabilistic method. For each cluster we show a random point such that if we move the points in each cluster to its corresponding random point, the expected Jaccard similarity is at least $\frac{\alpha^2}{8}$ times that of the initial Jaccard similarity. This implies that there exists a fixed (i.e., deterministic) set of points such that if we move the points in each cluster to its corresponding fixed point, the expected Jaccard similarity is at least $\frac{\alpha^2}{8}$ times that of the initial Jaccard similarity. Note that, for each cluster, we are selecting the optimum center, and hence this statement holds for our selected centers as well.

Let $E_{\text{init}}$ be the edge set corresponding to the clustering prior to Line \ref{line4.2}.
For each merged cluster we select the center of one of its initial clusters uniformly at random. Note that each merged cluster contains at most $2/\alpha$ initial clusters.
we select the center of each initial cluster with probability at least $\alpha/2$. This means that each edge that exists in $E\cap E_{\text{init}}$ exists after merging the clusters with probability at least $\alpha/2$. Hence, the expected number edges that exists after merging is at least $\frac{\alpha}{2} |E\cap E_{\text{init}}|$.

Moreover, the number of nodes in each initial cluster is at most $\frac{\alpha k}{2k}$ fraction of that of the merged cluster. Hence, the total number of edges increases by a factor of at most $2/\alpha$ after merging and moving the point to the random centers. Hence, the expected Jaccard similarity after the merge is at least
\begin{align*}
    \frac{\frac{\alpha}{2} |E\cap E_{\text{init}}|}{\frac{2}{\alpha}|E_{\text{init}}|+ |E|} & \geq  \frac{ \frac{\alpha}{2} |E\cap E_{\text{init}}|}{ 2 \frac{2}{\alpha} |E\cup E_{\text{init}}|} = \frac{\alpha^2}{8}\frac{  |E\cap E_{\text{init}}|}{  |E\cup E_{\text{init}}|},
\end{align*}
as claimed. This completes the proof of Theorem \ref{thm:main:improved}.
\end{proof}

\begin{figure*}
\centering
\subfigure[adult]{\includegraphics[width=0.31\textwidth,keepaspectratio]{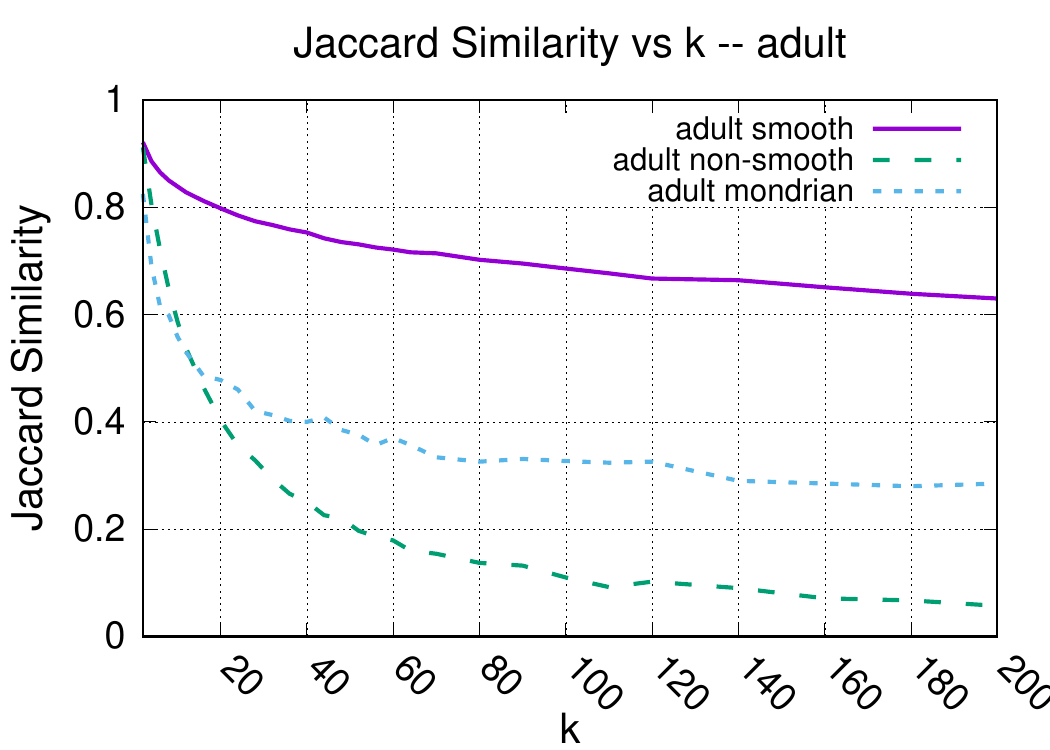}\label{fig:result-jaccard-adult}}
\subfigure[user-list]{\includegraphics[width=0.31\textwidth,keepaspectratio]{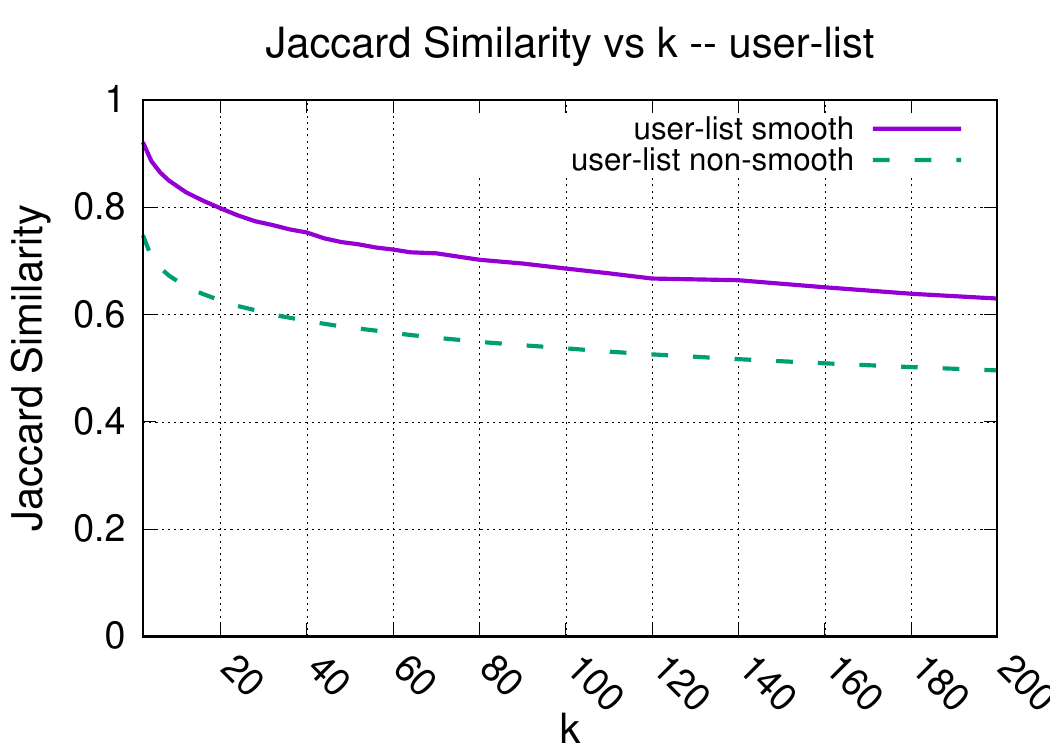}\label{fig:result-jaccard-user_list}}
\subfigure[$\epsilon$-edge DP]{\includegraphics[width=0.31\textwidth,keepaspectratio]{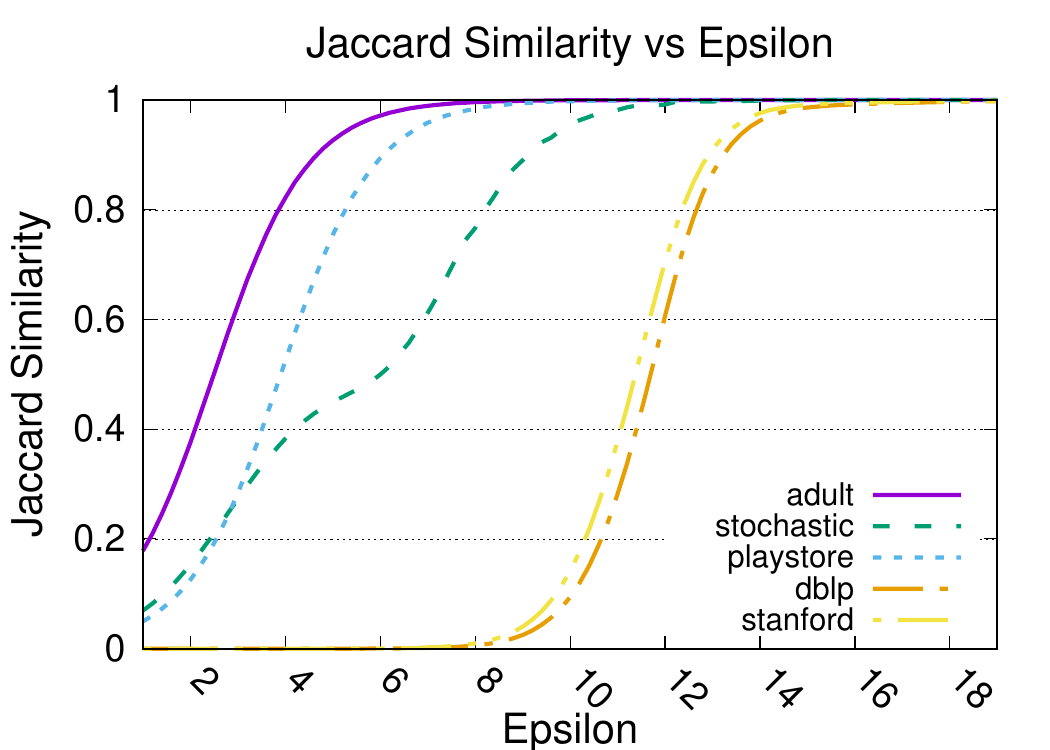}~\label{fig:result-eps}}
\caption{Mean Jaccard similarity for the various datasets and algorithms. \label{fig:result-jaccard}}.
\end{figure*}

\section{Experimental results}
\label{sec:exp}

We give a brief overview of the empirical performance of our algorithms. We give the full details of the setup as well as additional empirical results in~\refappendix{app:exp}.
 We will release an {\it open-source} version of our code by the camera-ready deadline.

\textbf{Datasets}
We used representative examples of sparse binary matrices (and bipartite graphs) of different origins and structural properties. The scales of the datasets are up to $>1B$ rows, $>100k$ columns, $> 10B$ entries (for our largest dataset) and the densities of the matrices range from $10^{-5}$ to $0.1$.
We used one synthetic dataset, four publicly-available real-world datasets as well as one large-scale proprietary dataset from a major internet company. These are as follows: 
{\bf stochastic} is generated from the stochastic block model described in Section~\ref{sec:stochastic};
{\bf adult}\footnote{\url{https://archive.ics.uci.edu/ml/datasets/adult}} and {\bf playstore}\footnote{\url{https://www.kaggle.com/lava18/google-play-store-apps}} consist of sparse binary matrices; 
{\bf dblp}~\cite{yang2015defining}, {\bf stanford}~\cite{leskovec2009community} consists of adjacency matrices of sparse bipartite graphs and, 
{\bf user-lists} is a proprietary dataset from a major internet company containing user-interest relationships.

{\bf Metrics} We measure utility using the {\it Jaccard similarity} between the set of the edges as defined in the paper as well as the number of {\it suppressed entries} and {\it created entries}. 

\textbf{Experimental infrastructure}
Our algorithm is implemented as a single-threaded C++ problem and is run on standard commodity hardware, with the exception of runs on the large-scale user-list dataset. For this dataset, we evaluated a simple heuristic to parallelize our algorithm.  (See \refappendix{app:exp} for more details.)

\textbf{Baselines and Algorithms}
As baselines we use the {\it Mondrian} anonymization algorithm~\cite{lefevre2006mondrian} implementation\footnote{\url{https://github.com/qiyuangong/Mondrian}} which enforces $k$-anonymity by suppression, as well as the randomized response algorithm of Section~\ref{sec:dp-comparison} which enforces $\epsilon$ edge- or node-differential privacy.  We also compare our algorithm for {\it smooth} $k$-anonymization with an additional baseline {\it non-smooth} (which uses a simple heuristic to obtain {\it (standard)} $k$-anonymity by suppression using the clusters obtained by our algorithm). Our algorithm is always run with $\alpha = 1/2$, which in practice gives good results for all datasets (we observe that values of $\alpha$ in $[1/8, 1/2]$ have similar results).

\begin{table}
\centering
\small
\begin{tabular}{llrrr}
\toprule
   dataset        & algorithm       &  Jaccard &  Supp. &  Created \\
\midrule
adult & mondrian &    59.9\% &               40.1\% &             0.0\% \\
           & non-smooth &    64.8\% &               35.2\% &             0.0\% \\
           & smooth &    85.0\% &                8.9\% &             7.2\% \\
playstore & mondrian &    51.2\% &               48.8\% &             0.0\% \\
           & non-smooth &    39.2\% &               60.8\% &             0.0\% \\
           & smooth &    66.1\% &               26.2\% &            11.6\% \\
user\_lists & non-smooth &    67.3\% &               32.7\% &             0.0\% \\
           & smooth &    71.0\% &               27.7\% &             1.9\% \\
\bottomrule
\end{tabular}
\caption{Average results for $k=8$ for various algorithms and dataset. Jac., S.E. and C.E. stand for, respectively, Jaccard similarity and fraction of suppressed entries and newly created entries (both normalized by the entries in the input dataset).\label{tab:results-k8}}
\end{table}

\textbf{Jaccard similarity vs $k$}
First, we evaluate the quality of our algorithm for smooth-$k$-anonymity for different $k$ values and we compare it with that of the (non-smooth) $k$-anonymity solution and {\it mondrian}. In Figures~\ref{fig:result-jaccard-adult} and \ref{fig:result-jaccard-user_list} we show a sample of plots of the mean Jaccard similarity for a given setting of the $k$ parameter for smooth-$k$-anonymity (solid line), non-smooth anonymity (dashed line) and mondrian (dotted). We were not able to run the mondrian algorithm on the larger datasets because, contrary to our algorithm, it scales with the size of the full $n \times m$ matrix size ($m$ number of columns) and it does not exploit the sparsity of the matrix. As expected, the Jaccard similarity decreases with increasing $k$, but at every $k$ level smooth-$k$-anonymity allows to obtain significantly better results than all baselines (in some cases even twice better). 
We report more detailed results in Table~\ref{tab:results-k8} for $k=8$. Notice that our smooth algorithm allows significantly higher jaccard similarity (and lower suppressed entries) for a small increase in created entries. For instance, in adult, the number of suppressed entries is decreased by $\sim 26\%$ with just a $\sim 7\%$ increase in added entries.

\begin{figure}
\centering
\includegraphics[width=0.75\textwidth,keepaspectratio]{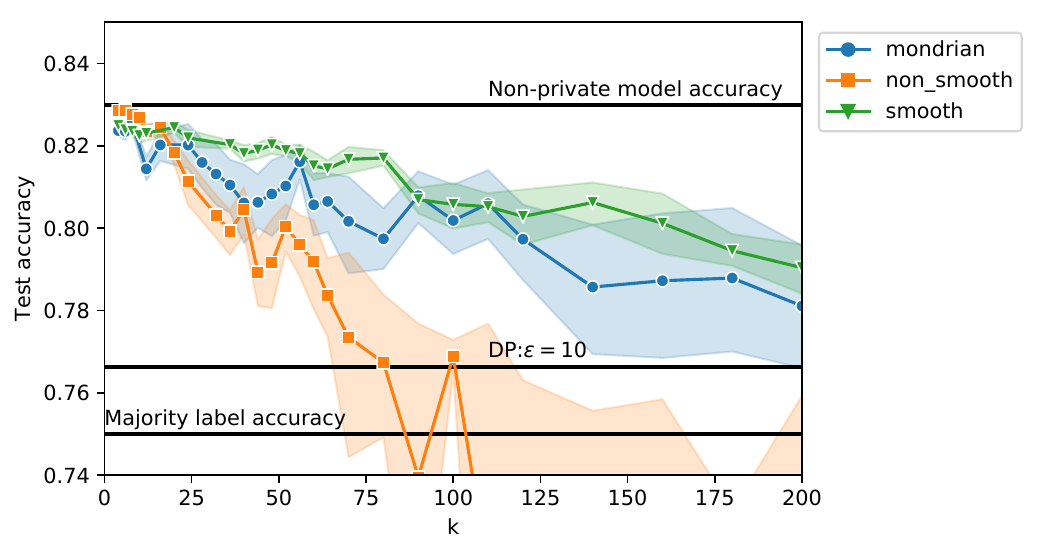}
\caption{Accuracy in learning task in anonymous data~\label{fig:ml-kanon}. We also include a baseline of using only the majority label as well as training a model without anonymity and a model that uses node differential privacy with $\epsilon=10$}
\end{figure}

\textbf{Differential privacy}
We now evaluate the Jaccard similarity obtained by the $\epsilon$-differentially private method.  We report the results in Figure~\ref{fig:result-eps}. Here we report results for the lower protection level of $\epsilon$-edge differential privacy, as $\epsilon$-node differential privacy protection generates results close to random outputs. 
As expected (see Section~\ref{sec:dp-comparison}) the sparser the dataset the worse the performance of differential privacy at parity of $\epsilon$. Notice how to get Jaccard similarity above $0.5$ in stanford or dblp, $\epsilon$ must be $10$ which is too large to provide strong guarantees. 
We can use these results to compare $k$-anonymity and $\epsilon$-differential from their a standpoint.  We observe that depending on the dataset an anonymity of $k=16$ might require an $\epsilon$ as large as $11$ to obtain the same utility. 

\textbf{Learning from anonymous data}
Finally, we report results on using the anonymized datasets in a downstream machine learning task. We use the anonymized version of the adult dataset to learn a classifier for the standard classification task of predicting whether an adult's income is $\geq\$50k$ per year. The results are reported in Figure~\ref{fig:ml-kanon}. 
Notice our algorithm performs better (or on par) with the best baseline (mondrian). We observe (see \refappendix{app:exp}) that smooth performs significantly better than the $\epsilon$-node differentially private algorithm with $\epsilon = 10$ even for $k=200$, mirroring the degradation seen in the Jaccard similarity metric. 

\section{Conclusion}
We presented a new notion of anonymity which relaxes standard $k$-anonymity and allows us to obtain better guarantees and improved results in downstream machine learning tasks. Our algorithms reduce the anonymization problem to clustering with lower bounds and require non-trivial analysis to prove approximation guarantees. Many interesting questions remain, including  strengthening our bounds for smooth-$k$-anonymity and better understanding the interplay between the various notions of privacy.

\section{Acknowledgement}
The authors are thankful to Sergei Vassilvitskii for his involvement in early stages of this paper and his useful discussions and suggestions.
\bibliographystyle{ACM-Reference-Format}
\bibliography{references}

\ifappendix
\appendix

\clearpage
\newpage
\section{Technical lemma}
\begin{lemma}
\label{lem:ratio}
Let $a, b, x >0$  be such that  $x < \frac{b-a}{2}$ then
\begin{equation}
    \label{eq:ratio}
    \frac{a - x}{b-x} < \frac{a}{b} + 2 x.
\end{equation}
\end{lemma}
\begin{proof}
From the condition on $x$ it follows that $b - x>0$. Therefore rearranging terms in \eqref{eq:ratio} yields the equivalent inequality:
\begin{equation*}
    ab - bx < ab - ax + 2bx  - 2 x^2.
\end{equation*}
This inequality holds if and only if $x < \frac{b-a}{2}.$
\end{proof}

\section{Omitted proofs from Section~\ref{sec:dp-comparison}}
\subsection{Randomized response is $\epsilon$-differentially private}\label{app:dp}

Consider the following procedure. For each element, with probability $1-\frac{2}{e^{\epsilon}+1}$ we report the true value, otherwise with probability $\frac{1}{e^{\epsilon}+1}$ we report $0$ and with probability $\frac{1}{e^{\epsilon}+1}$ we report $1$. We refer to this procedure as $\epsilon$-DP. Next theorem shows that this procedure is $\epsilon$-differentially private.
\begin{theorem}
$\epsilon$-DP preserves $\epsilon$ edge differential privacy.
\end{theorem}
\begin{proof}
Let $G$ and $G'$ be two arbitrary graphs that only differ in one edge $e$, and let $G_{\epsilon}$ and $G'_{\epsilon}$ be the outcomes of $\epsilon$-DP on $G$ and $G'$ respectively. Let $S$ be a collection of graphs, and let $H$ be an arbitrary graph in $S$.First we consider the case that $e\in G$ and $e\in H$. Then, we consider the case that $e\notin G$ and $e\in H$. The same argument holds for the other two possibilities.

Let $P^{-e}(G_{\epsilon},H)$ be the probability that $G_{\epsilon} - \{e\} = H- \{e\}$, i.e., ignoring edge $e$ all of their other edges match. Similarly, let $P^{-e}(G'_{\epsilon},H)$ be the probability that $G'_{\epsilon} - \{e\} = H- \{e\}$. Note that $G$ and $G'$ only defer on edge $e$. Hence, we have $P^{-e}(G_{\epsilon},H)=P^{-e}(G'_{\epsilon},H)$. We have
\begin{align*}
    \textbf{Pr}(G_{\epsilon} \in S) &= \sum_{H\in S} P(G_{\epsilon} = H)\\ 
    &=\sum_{H\in S} P^{-e}(G_{\epsilon},H) \textbf{Pr}(e \in G_{\epsilon})\\
    &=\sum_{H\in S} P^{-e}(G_{\epsilon},H) \big(1-\frac{1}{e^{\epsilon}+1}\big)\\
    &=\big(\frac{e^{\epsilon}}{e^{\epsilon}+1}\big)\sum_{H\in S} P^{-e}(G_{\epsilon},H) .
\end{align*}
Similarly, we have
\begin{align*}
    \textbf{Pr}(G'_{\epsilon} \in S) &= \sum_{H\in S} P(G'_{\epsilon} = H)\\ 
    &=\sum_{H\in S} P^{-e}(G'_{\epsilon},H) \textbf{Pr}(e \in G'_{\epsilon})\\
    &=\sum_{H\in S} P^{-e}(G'_{\epsilon},H) \big(\frac{1}{e^{\epsilon}+1}\big)\\
    &= \big(\frac{1}{e^{\epsilon}+1}\big) \sum_{H\in S} P^{-e}(G'_{\epsilon},H)\\
    &= \big(\frac{1}{e^{\epsilon}+1}\big) \sum_{H\in S} P^{-e}(G_{\epsilon},H).
\end{align*}
Hence, we have $\textbf{Pr}(G_{\epsilon} \in S) = e^ {\epsilon} \textbf{Pr}(G'_{\epsilon} \in S)$ as desired.

Now we consider the case that $e\notin G$ and $e\in H$. Similar argument to the previous case implies $  \textbf{Pr}(G_{\epsilon} \in S) = \big(\frac{1}{e^{\epsilon}+1}\big)\sum_{H\in S} P^{-e}(G_{\epsilon},H)$ and $\textbf{Pr}(G'_{\epsilon} \in S) = \big(\frac{e^{\epsilon}}{e^{\epsilon}+1}\big) \sum_{H\in S} P^{-e}(G_{\epsilon},H)$. This implies that $e^ {\epsilon} \textbf{Pr}(G_{\epsilon} \in S) =  \textbf{Pr}(G'_{\epsilon} \in S)$. For $\epsilon\geq 1$ we have $e^{\epsilon}\geq 1$. Hence we have $ \textbf{Pr}(G_{\epsilon} \in S) \leq e^ {\epsilon} \textbf{Pr}(G'_{\epsilon} \in S)$ as desired.
\end{proof}

\subsection{Bound on utility of the randomized response method}\label{app:upperbound-eps-proof}
\begin{theorem}[Restatement of Theorem~\ref{prop:upperbound-eps}]
Let $G=(U \cup F, E) \in \mathbb{G}$ be a graph, $\delta > 0$. Let $\cM$ be a mechanism that generates a graph according to Algorithm~\ref{alg:randomized_response} with $p = \frac{2}{1 + e^\epsilon}$.
Let $Q = |U||F|$ and $C(\delta) = \frac{\log(2/\delta)}{2}$. If $\frac{p}{4}\geq \sqrt{\frac{C(\delta)}{Q}}$, then with probability at least $1 - \delta$:
$$J(\cM(G), G) \leq \frac{1 - \frac{1}{e^{\epsilon} + 1}}{1 + \frac{1}{e^{\epsilon} + 1}\frac{1 - \lambda}{\lambda}} + 2\sqrt{\frac{C(\delta)}{Q}},$$
where $\lambda := \lambda(G)$.
\end{theorem}
\begin{proof}
Let $E'$ denote the edge set of the output graph. By definition of the mechanism, each edge on $E$ is preserved with probability $1 - p + \frac{1}{2} p = 1-\frac{p}{2}$. Thus, we have $|E \cap E'| = \sum_{i=1}^{|E|} Y_i$ where $Y_i$ is distributed Bin($1 - \frac{p}{2}$). By Hoeffding's inequality we know with probability at least $1 - \frac{\delta}{2}$:
\begin{align*}
    |E\cap E'| &\leq \left(1 - \frac{p}{2}\right)|E| + \sqrt{|E|\frac{\log(2/\delta)}{2}} \\
    &= \left(1 - \frac{p}{2}\right) \lambda Q +\sqrt{\lambda Q C(\delta)}
\end{align*}
Similarly an edge not in $E$ is added to $E'$ with probability $\frac{p}{2}$. Thus we have $|E \cup E'| =  \sum_{j=1}^{Q - |E|} Z_j$ where $Z_j\sim \textrm{Bin}\left(\frac{p}{2}\right)$. Thus, by Hoeffding's inequality we have that with probability at least $1 - \frac{\delta}{2}$
\begin{align*}
    |E\cup E'| &\geq |E| + \left(\frac{p}{2}\right)(Q - |E|) - \sqrt{(Q - |E|) \frac{\log(2/\delta)}{2}} \\
    &=\lambda Q + \left(\frac{p}{2}\right)Q( 1 - \lambda)  - \sqrt{Q(1 - \lambda)C(\delta)} 
\end{align*}
By the union bound, and rearranging terms we thus have that with probability at least $1 - \delta$
\begin{align*}
J(\cM(G), G) &\leq\frac{(1 - \frac{p}{2})\lambda + \sqrt{ \lambda \frac{C(\delta)}{Q}}}{\lambda + \frac{p}{2}(1 - \lambda) -\sqrt{ (1 - \lambda)\frac{C(\delta)}{Q}}} \\
&=\frac{(1 - \frac{p}{2})\lambda + \sqrt{ \frac{C(\delta)}{Q}}}{\lambda + \frac{p}{2}(1 - \lambda) -\sqrt{\frac{C(\delta)}{Q}}} \\
& \leq \frac{(1 - \frac{p}{2})\lambda}{\lambda + \frac{p}{2}(1 - \lambda)} +  2\sqrt{ \frac{C(\delta)}{Q}},
\end{align*}
where we have used Lemma~\ref{lem:ratio} in the Appendix for the last inequality.  The results follow by dividing the numerator and denominator by $\lambda$ and using the definition of $p$.
\end{proof}

\subsection{Utility of k-anonymity by suppression in the stochastic block model}
\label{app:thm:sbm}
\begin{theorem}[Restatement of Theorem~\ref{thm:sbm}]
Let $G$ be a graph generated by the stochastic block model with parameters $r,s,\alpha,\beta$. Let $k \geq \frac{2\log n}{\log  1 / q} $. With probability $99\%$, any $k$-anonymous subgraph of $G$ contains at most $\frac{2\log n+10}{\log  1 / q} n  \in \tilde{O}(n)$ edges. 
\end{theorem}
\begin{proof}
Let $t=\frac{2\log n+10}{\log \frac {1} {q}}$. 
Fix a subset of size $k$ of vertices in the first part and a subset of size $t$ in the second part. Recall that the probability that we have an edge between a pair of vertices is either $p$ or $q$, where $p\leq q$, and each edge exists independently. Hence, the probability that all of the edges between a certain set of vertices of size $k$ and a certain set of vertices of size $t$ exist is at most $q^{kt}$. On the other hand there are ${k \choose n} \cdot {t \choose n}$ such a pair of sets. Therefore, by union bound, the probability that there exists a set of vertices of size $k$ and a set of vertices of size $t$ in $G$ that forms a complete bipartite graph is at most 
\begin{align*}
    q^{kt} \cdot {k \choose n} \cdot {t \choose n} &\leq q^{kt} n^k n^t \\
    &= \exp \big(kt\log q + k \log n + t \log n\big)\\
    &= \exp \big(-kt\log \frac 1 q + k \log n + t \log n\big)\\
    &= \exp \big(  (k \log n -k\frac t 2 \log \frac 1 q ) 
    + (t \log n-\frac k 2 t\log \frac 1 q)\big)\\
    &\leq \exp \big(  (k \log n -k\frac t 2 \log \frac 1 q )\big)\\
    &= \exp  ( 5 k  )\\
    &\leq 0.01.
\end{align*}
Where the inequalities follow by the bounds on $k$ and $t$. 
This implies that with probability $99\%$ there is no vertex with degree more than $t$ in a $k$-anonymous subgraph of $G$. Hence the total number of edges in such a subgraph is at most $t n  = \frac{2\log n+10}{\log \frac 1 q} n  \in \tilde{O}(n)$.
\end{proof}

\section{Omitted proofs from Section~\ref{sec:technical}}
\label{app:sec:technical}

\subsection{Lemma~\ref{lm:Jaccard2Count}}\label{app:lm:Jaccard2Count}
\begin{lemma}[Restatement of Lemma~\ref{lm:Jaccard2Count}]
Assume $J(E,E_{\opt}) \geq 1-\phi$. We have $\sum_{u} \Delta_u^{\opt} \leq \frac{2\phi}{1-\phi} |E|$.
\end{lemma}
\begin{proof}
Since $J(E,E_{\opt}) \geq 1-\phi$, we have $\frac{|E\cap E_{\opt}|}{|E\cup E_{\opt}|}\geq 1-\phi$. Thus $\frac{|E \oplus E_{\opt}|}{|E\cup E_{\opt}|}\leq \phi$, which gives us $|E \oplus E_{\opt}| \leq \phi |E\cup E_{\opt}|$. On the other hand $\frac{|E\cap E_{\opt}|}{|E\cup E_{\opt}|}\geq 1-\phi$ implies that $|E\cup E_{\opt}| \leq \frac{1}{1-\phi} |E\cap E_{\opt}| \leq \frac{1}{1-\phi}|E|$. Therefore we have
\begin{align*}
    \sum_{u} \Delta_u^{\opt} = 
    2 |E \oplus E_{\opt}| 
    \leq 2\phi |E\cup E_{\opt}| 
    \leq \frac{2\phi}{1-\phi} |E|.
\end{align*}
\end{proof}

\subsection{Lemma~\ref{lm:Count2Jaccard}}\label{app:lm:Count2Jaccard}
\begin{lemma}[Restatement of Lemma~\ref{lm:Count2Jaccard}]
Assume $\sum_{u} \Delta_u^{\alg} \leq \phi' |E|$. We have $J(E,E_{\alg}) \geq 1-\frac{\phi'}{2}$.
\end{lemma}
\begin{proof}
 Note that $\sum_{u} \Delta_u^{\alg} \leq \phi' |E|$ means $ 2 |E \oplus E_{\opt}| \leq \phi' |E|$. Hence we have 
 \begin{align*}
2\big(|E\cup E_{\alg}| - |E\cap E_{\alg}|\big)\leq \phi' |E| \leq \phi' |E\cup E_{\alg}|.
 \end{align*}
By some rearrangement we have 
$\frac{|E\cap E_{\alg}|}{|E\cup E_{\alg}|} \geq 1-\frac{\phi'}{2}$,
 which means $J(E,E_{\alg}) \geq 1-\frac{\phi'}{2}$ as desired.
\end{proof}

\section{Additional material from the experimental section}
\label{app:exp}

\paragraph{Description of the datasets}
We now report a more detailed description of the datasets used. Some key properties of the datasets are shown in Table~\ref{tab:datasets}.

\begin{table}[ht] 
\small
\centering 
\caption{Datasets used in the analysis} \label{tab:datasets}
\begin{tabular}{l|llll}
\hline
Name & Rows  & Columns & Entries & Density \\
\hline
stochastic & $1024$ & $1024$ & $\sim 62{,}000$ & $\sim 0.06$\\
adult & $32{,}561$ & $102$ & $260{,}488$ & $0.078$\\
playstore  & $8{,}672$ & $521$ & $86{,}720$& $0.019$\\
dblp  & $317{,}080$ & $317{,}080$ & $2{,}099{,}732$ & $2.09 \times 10^{-5}$\\
stanford & $281{,}731$ & $261{,}588$ & $2{,}312{,}497$ & $3.14 \times 10^{-5}$\\
user-lists & $>1B$  & $>100K$ & $>10B$ & $<10^{-5}$\\
\hline
\end{tabular}
\end{table}

{\bf stochastic} is generated from the stochastic block model described in Section~\ref{sec:stochastic} with $n = 1024$ rows and columns, and parameters $s=64$, $q = 0.8$, $p = 0.01$.
{\bf adult} is a dataset from UCI~\footnote{\url{https://archive.ics.uci.edu/ml/datasets/adult}} from the US census used in many anonymization and differential privacy papers~\cite{lefevre2006mondrian}. We use all categorical features represented as a binary vector.
{\bf playstore}\footnote{\url{https://www.kaggle.com/lava18/google-play-store-apps}} consists of a bipartite graph representing features of user's devices in Google Playstore. The items $F$ represent different features that the device has. 
{\bf dblp}~\cite{yang2015defining} consists of a co-authorship graph. In graph adjacency form, rows and columns represent authors in computer science and an entry $(i,j)$ is 1 iff $i$ and $j$ co-authored a paper.
{\bf stanford}~\cite{leskovec2009community} consists of a directed web graph, where rows and columns represent web pages crawled from the Stanford website. An entry $(i,j)$ is $1$ iff page $i$ has a hyperlink to page $j$.
{\bf user-lists} is a proprietary dataset from a major internet company containing user-interest relationships. Each row represents a user, and the columns represent the inclusion in a given marketing interest list.

\paragraph{Metrics}
Our utility goal is to preserve as accurately as possible the input, to do so we use the {\it Jaccard similarity} between the set of the edges as utility metric as defined in the paper. To measure the distortion incurred by the data we also use other standard anonymity quality metrics, including the number of {\it suppressed entries} and {\it created entries}. 
We note that for binary data other metrics such as the {\it normalized certainty penalty}~\cite{lefevre2006mondrian} used for hierarchical features are not relevant.

\paragraph{Experimental infrastructure}
We implemented our algorithm using C++. For all execution of an experiment (except for the large-scale user-list dataset) we use a single-thread program run on standard commodity hardware. All experiments terminated within $3$ hours. We repeat each run of an algorithm with a given setting $10$ times and report mean and deviation metrics. 
For our user-list dataset with more than a billion rows, we evaluated a simple heuristic to parallelize the computation of our algorithm. We split the dataset in chunks of order of $10^5$ rows by using a Locality Sentitive Hashing technique (LSH)~\cite{wang2014hashing}. More precisely, for each row we compute $8$ independent min hash LSHs (over the set of columns with $1$) and sort the rows of the dataset by the lexicographical order of their min hashes. This is similar to the shingle sorting technique~\cite{chierichetti2009compressing}. We then obtain the chunks by approximately dividing this sequence of elements in consecutive intervals of size $10^5$. Then we ran our algorithm independently on each chunk of the dataset and combined the solution. Our results (reported in Fig~\ref{fig:result-jaccard-user_list}) show that this simple heuristic is quite effective.

\paragraph{Baselines}
We compare our algorithm with the well-known {\it mondrian} anonymization algorithm~\cite{lefevre2006mondrian}, using an open source implementation\footnote{\url{https://github.com/qiyuangong/Mondrian}} which enforces k-anonymity by suppression. The algorithm doesn't have parameters to tune (expect $k$).

Our algorithm has theoretical guarantees for {\it smooth} k-anonymization, but it can be easily used as an heuristic for (standard) k-anonymity by suppression (in the last step, the clusters are simply anonymized by computing the intersection of all binary vectors), so we report results for {\it non-smooth} k-anonymity using our algorithm as an additional baseline.

\paragraph{Implementation details}
For efficiency we made small changes to the implementation of the algorithm in Section~\ref{alg:improved} that we now summarize.
In step (2) we slightly modify the formula for obtaining the best empirical results. Similarly to the bicriteria algorithm of~\cite{svitkina2010lower}, we add a parameter  $\beta > 1$, set $\alpha = 1/\beta$, and define for a point $u_i$ the cost of the facility as $\frac{2\alpha}{1-\alpha} \sum_{u'\in U^{\beta k}_i} \dist(u',u_i)$, where $U^{\beta k}_i$ is the set of $\beta k$ closest point to $i$. In our experiments, $\beta=2$ (i.e. $\alpha = 1/2$) was the best performing setting and we report experiments with that value only (results with $2\le  \beta \le 8$ are very similar).
In step 3, we solve the facility location problem efficiently by running the well-known Meyerson's algorithm~\cite{meyerson2001online} 10 different times on a random ordering of the data points and selecting the best solution (in this setting the algorithm is known to be a constant factor approximation~\cite{fotakis2011online} in expectation).
Finally, in steps 4 and 5, we replace the procedure described in the paper, with a simpler one where we simply close facilities that have fewer than $k$ clients in arbitrary order and reassign all clients to the nearest open one until all clusters are large enough.

\begin{table}
	\centering
\small
\begin{tabular}{llrrr}
\toprule
           &        &  Jac. &  S.E. &  C.E. \\
dataset & algorithm &          &                     &                  \\
\midrule
stochastic & mondrian &     2.2\% &               97.8\% &             0.0\% \\
           & non-smooth &    16.4\% &               83.6\% &             0.0\% \\
           & smooth &    68.1\% &               17.7\% &            21.0\% \\
adult & mondrian &    59.9\% &               40.1\% &             0.0\% \\
           & non-smooth &    64.8\% &               35.2\% &             0.0\% \\
           & smooth &    85.0\% &                8.9\% &             7.2\% \\
playstore & mondrian &    51.2\% &               48.8\% &             0.0\% \\
           & non-smooth &    39.2\% &               60.8\% &             0.0\% \\
           & smooth &    66.1\% &               26.2\% &            11.6\% \\
dblp & non-smooth &     4.3\% &               95.7\% &             0.0\% \\
           & smooth &     7.4\% &               92.5\% &             1.4\% \\
stanford & non-smooth &    48.0\% &               52.0\% &             0.0\% \\
           & smooth &    52.6\% &               46.1\% &             2.4\% \\

user\_lists & non-smooth &    67.3\% &               32.7\% &             0.0\% \\
           & smooth &    71.0\% &               27.7\% &             1.9\% \\
\bottomrule
\end{tabular}
\caption{Extended version of Table~\ref{tab:results-k8}. Average results for $k=8$ for various algorithms and dataset. Jac., S.E. and C.E. stand for, respectively, Jaccard similarity and fraction of suppressed entries and newly created entries (both normalized by the entries in the input dataset).\label{tab:results-k8-extended}}
\end{table}

\begin{figure*}
\centering
\subfigure[stochastic]{\includegraphics[width=0.4\textwidth,keepaspectratio]{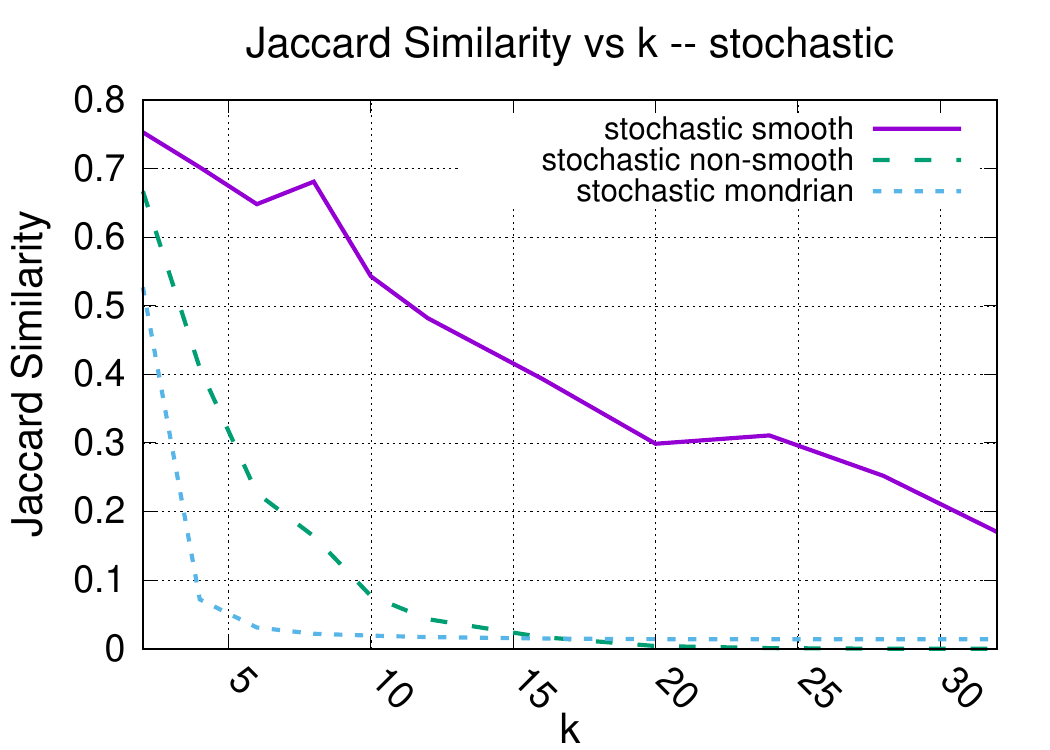}\label{fig:result-jaccard-stochastic-app}}
\subfigure[playstore]{\includegraphics[width=0.4\textwidth,keepaspectratio]{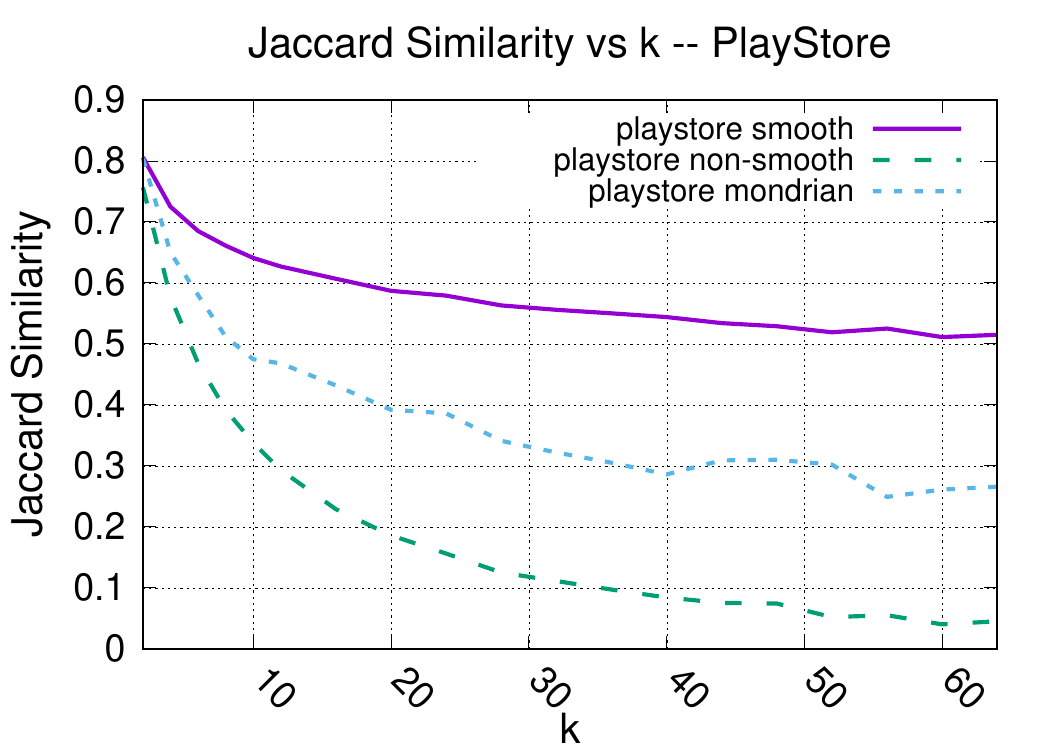}\label{fig:result-jaccard-playstore-app}}
\subfigure[dblp]{\includegraphics[width=0.4\textwidth,keepaspectratio]{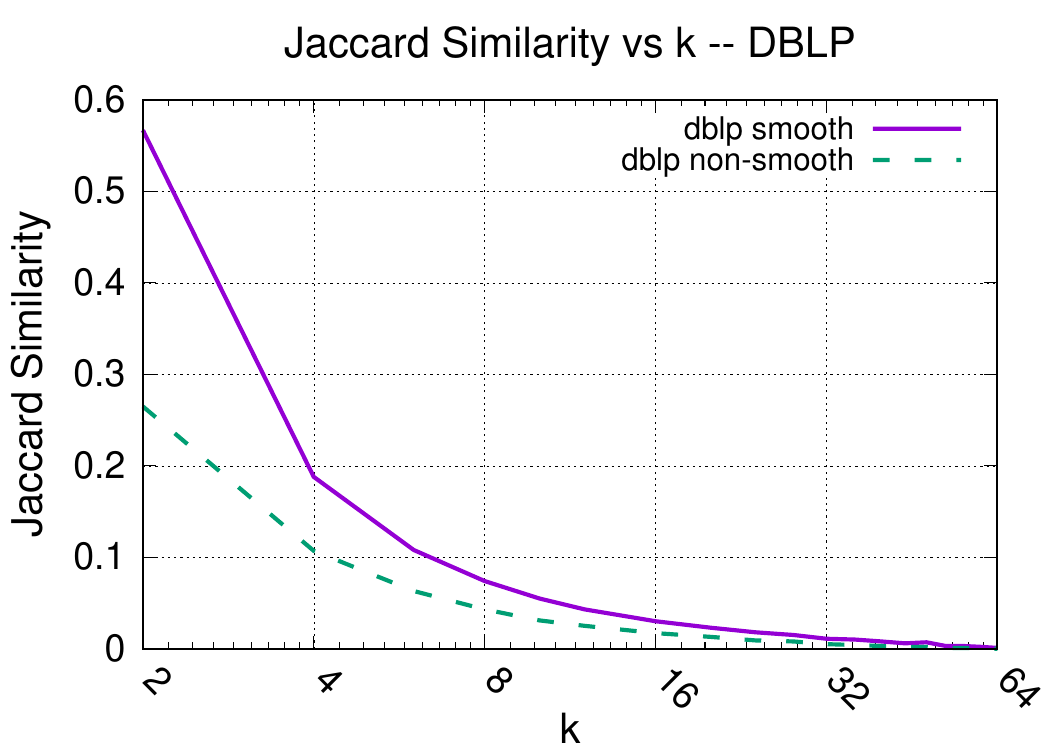}\label{fig:result-jaccard-dblp-app}}
\subfigure[stanford]{\includegraphics[width=0.4\textwidth,keepaspectratio]{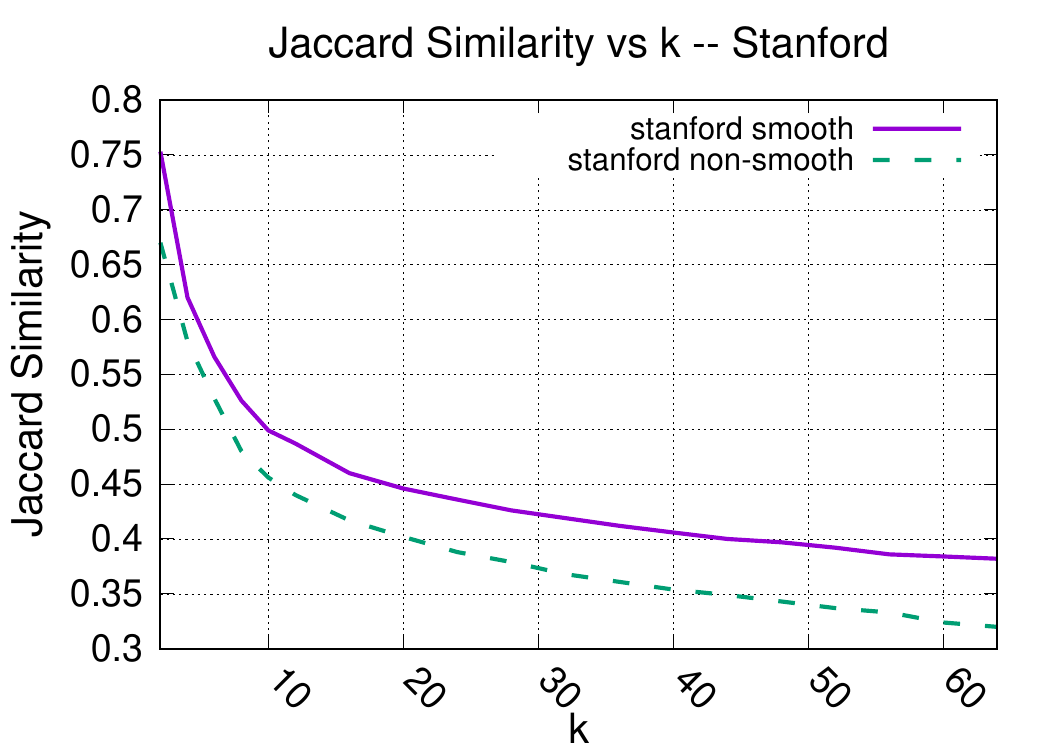}\label{fig:result-jaccard-stanford-app}}
\caption{Mean Jaccard similarity vs k for additional datasets. \label{fig:result-jaccard-app}}.
\end{figure*}

\paragraph{Jaccard similarity vs $k$}
First, we evaluate the quality of our algorithm for smooth $k$-anonymity for different $k$ values and we compare it with that of the (non-smooth) k-anonymity solution and {\it mondrian}. In Figures~\ref{fig:result-jaccard-stochastic-app},~\ref{fig:result-jaccard-playstore-app},~\ref{fig:result-jaccard-dblp-app}, ~\ref{fig:result-jaccard-stanford-app} as well as Figures~\ref{fig:result-jaccard-adult} and~\ref{fig:result-jaccard-user_list} we plot the mean Jaccard similarity for a given setting of the $k$ parameter for our datasets for both smooth k-anonymity (solid line), non-smooth anonymity (dashed line) and mondrian (dotted). We were not able to run the mondrian algorithm on the larger datasets (dblp, stanford, user-list) as this algorithm (contrary to ours) scales with the size of the full $n \times m$ matrix where $m$ is the number of columns of the dataset (while our algorithm can exploit the sparsity of the matrix).

From the pictures, it is possible to observe that, as expected, the Jaccard similarity decreases with increasing $k$, but at every $k$ level smooth k-anonymity allows to obtain significantly better results than all baselines. For instance, for playstore we observe $2\times$ higher Jaccard similarity for $k=32$ than with the best baseline. Among the baselines, interestingly we see that non-smooth anonymity despite using our algorithm which is not optimized for anonymity by suppression can perform better than mondrian which is optimize for that task, in particular for lower k values. For large k values in Figure~\ref{fig:result-jaccard-adult} mondrian performs better but still is outperformed by smooth anonymity.

The level of Jaccard similarity at a given $k$ is, as expected, dependent on the structure of the graph. We observe for instance higher similarity at large $k$ values for stochastic, adult, playstore dataset than dblp. This can be easily explained by the nature of the datasets: for adult, stochastic, playstore, the rows represent comparably dense features; while for dblp each row represent a quite sparse and somewhat unique feature set (i.e., the co-authors of a researcher). Interestingly, the  algorithm performs much better on the stanford dataset than the dblp one, this can be explained with previous observations on social networks showing higher entropy than web graphs~\cite{chierichetti2009compressing}. Notice also how our heuristic algorithm for the large-scale dataset obtains performances (Figure~\ref{fig:result-jaccard-user_list}) comparable to the best datasets.

We observe that the results are very concentrated, in fact, the maximum standard deviation observed for any $k$ value in $[2,64]$ for our smooth (non-smooth) k-anonymity algorithm is 	$0.7\%, 1.3\%$, $0.5\%$, $1.3\%$ ($1.8\%$, $2.2\%$, $0.2\%$, $0.8\%$) for adult, playstore, dblp, stanford, respectively.

\paragraph{Trade-off between suppressed entries and created entries}
We report an extended version of Table~\ref{tab:results-k8} in Table~\ref{tab:results-k8-extended}.

\begin{figure}
\centering
\includegraphics[width=0.50\textwidth,keepaspectratio]{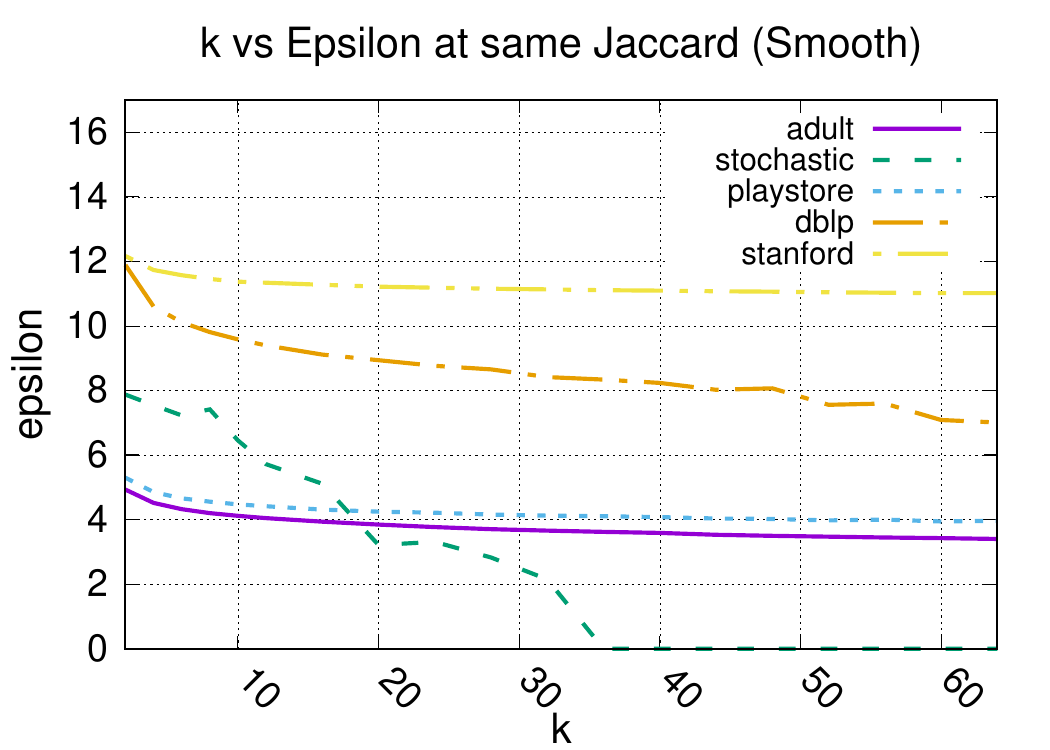}
\caption{$k$ anonymity parameter vs $\epsilon$-edge DP parameter at same level of mean Jaccard Similarity for the various datasets\label{fig:result-k-vs-eps}}.
\end{figure}

\paragraph{Additional results for differential privacy}
All results in our experiments with differential privacy in Figure~\ref{fig:result-eps} are very concentrated around the mean with maximum standard deviation for any $\epsilon$ of $1.02\%$, $0.09\%$, $0.15\%$, $0.05\%$, $0.04\%$ for stochastic, adult, playstore, dblp and stanford, respectively.

\paragraph{Comparison between k-anonymity and differential privacy}
We can use the results of the previous evaluation to compare the two approaches. Notice that $k$-anonymity and $\epsilon$-differential privacy provide mathematically incomparable privacy guarantees; but we can compare the utility of both  for different settings of their parameters $k$ or $\epsilon$. 

This is what we do in Figure~\ref{fig:result-k-vs-eps} where we report for each dataset and value of $k$ for smooth $k$-anonymity and $\epsilon$ parameter than will result in the same average Jaccard similarity for edge differential privacy protection. Notice how, depending on the dataset, an anonymity level of $k=16$ might require an $\epsilon$ as large as $11$ to obtain the same utility even for the lower level of guarantee of edge differential privacy. For the denser datasets like adult and playstore, values of anonymity between $k=4$ and $k=64$ correspond to the performance of $\sim 4 \le \epsilon \le \sim 5$.

\begin{figure}
\centering
\subfigure[DP]{\includegraphics[width=0.5\textwidth,keepaspectratio]{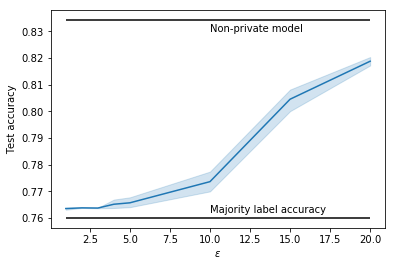}\label{fig:ml-dp}}
\caption{Accuracy in learning task in anonymous data~\label{fig:ml}.}
\end{figure}

\paragraph{Learning from anonymous data}
Finally, we report results on using the anonymized datasets in a downstream machine learning task. For the adult dataset we train a neural net classifier for the standard classification task for this dataset of predicting whether the income of the entry point is $>=\$50k$ per year. For all methods, we only anonymize the features not the labels and we train the model on the anonymized features in output from the algorithms and report the prediction on the original test dataset. We use a 1-layer, 10-hidden unit network with RELU activation nodes. 

The results are reported in Figure~\ref{fig:ml-kanon} for the anonymization methods and in Figure~\ref{fig:ml-dp} for the $\epsilon$-node differentially private method. The shades represent the $95\%$ confidence interval. 
Notice that our algorithm reports results better (or on par) with the best baseline (mondrian) and significantly better than the $\epsilon$-node differentially private algorithm with $\epsilon = 10$ even for $k=200$. We also observe that our smooth algorithm has lower variance in the output than the mondrian one. This confirms that our algorithm outperforms all baselines also in producing outputs useful for machine learning analysis.

\fi

\end{document}